\documentclass[journal]{IEEEtran}

\ifCLASSINFOpdf
\else
\fi

\usepackage{amsfonts}
\usepackage{amsmath}
\usepackage{amssymb}
\usepackage{bbding}
\usepackage{graphicx}
\usepackage{subfigure}
\usepackage{algorithm, algorithmic}
\usepackage{booktabs}
\usepackage{multirow}
\usepackage{diagbox}
\usepackage{color}
\usepackage{cite}
\usepackage{hyperref}
\usepackage{cleveref}

\newtheorem{proposition}{Proposition}
\newtheorem{proof}{Proof}

\DeclareMathOperator*{\argmin}{argmin}

\begin{document}

\title{A Novel Two-stage Design Scheme of Equalizers for Uplink FBMC/OQAM-based Massive MIMO Systems}

\author{Yuhao~Qi,
        Jian~Dang,~\IEEEmembership{Senior~Member,~IEEE,}
        Zaichen~Zhang,~\IEEEmembership{Senior~Member,~IEEE,}
        Liang~Wu,~\IEEEmembership{Member,~IEEE,}
        and~Yongpeng~Wu,~\IEEEmembership{Senior~Member,~IEEE}
\thanks{This work was supported by the NSFC projects (61971136, 61960206005, and 61803211), Jiangsu NSF project (No. BK20191261), Universities Natural Science Research Project of Jiangsu Province, China (No. 19KJB510048), Social Development Projects of Jiangsu Science and Technology Department (No. BE201874), Zhejiang Lab (No. 2019LC0AB02), Zhishan Youth Scholar Program of SEU, and Postgraduate Research $\&$ Practice Innovation Program of Jiangsu Province (No. KYCX21$\_$0106).}
\thanks{Y. Qi, J. Dang, Z. Zhang, and L. Wu are with the National Mobile Communications Research Laboratory, Southeast University, Nanjing 210096, China. Z. Zhang is also with the Purple Mountain Laboratories, Nanjing 211111, China. Y. Wu is with the Department of Electronic Engineering, Shanghai Jiao Tong University, Minhang 200240, China (e-mail: qiyuhao@seu.edu.cn; dangjian@seu.edu.cn; zczhang@seu.edu.cn; wuliang@seu.edu.cn; yongpeng.wu@sjtu.edu.cn).}
\thanks{Corresponding author: J. Dang (dangjian@seu.edu.cn)}}

%

\maketitle

\begin{abstract}
The self-equalization property has raised great concern in the combination of offset-quadratic-amplitude-modulation-based filter bank multi-carrier (FBMC/OQAM) and massive multiple-input multiple-output (MIMO) system, which enables to decrease the interference brought by the highly frequency-selective channels as the number of base station (BS) antennas increases. However, existing works show that there remains residual interference after single-tap equalization even with infinite number of BS antennas, leading to a limitation of achievable signal-to-interference-plus-noise ratio (SINR) performance. In this paper, we propose a two-stage design scheme of equalizers to remove the above limitation. In the first stage, we design high-rate equalizers working before FBMC demodulation to avoid the potential loss of channel information obtained at the BS. In the second stage, we transform the high-rate equalizers into the
low-rate equalizers after FBMC demodulation to reduce the implementation complexity. Compared with prior works, the proposed scheme has affordable complexity under massive MIMO and only requires instantaneous channel state information (CSI) without statistical CSI and additional equalizers. Simulation results show that the scheme can bring improved SINR performance. Moreover, even with finite number of BS antennas, the interference brought by the channels can be almost eliminated.
\end{abstract}


\IEEEpeerreviewmaketitle

\section{Introduction}
\IEEEPARstart{M}{ASSIVE} multiple-input multiple-output (MIMO) \cite{TLM10,FREA13,LGAA14} is a promising fifth generation (5G) wireless access technology that can satisfy stringent requirements, such as higher throughput and better accommodation for multi-user systems. By equipping with a large number of antennas at the base station (BS), the signals can be processed coherently over the antennas. Thus, the impacts of uncorrelated noise and multi-user interference can be arbitrarily small as the number of antennas increases to infinity. Other benefits of massive MIMO include extensive use of inexpensive low-power components, reduced latency, simplification of the MAC layer, and robustness against intentional jamming \cite{EOFT14}.

On the other hand, waveform design is also an important aspect to promote the performance of communication system. As one of the typical waveforms, orthogonal frequency division multiplexing (OFDM) \cite{OFDM90} has been widely adopted due to its high data rate transmission capability and robustness to multi-path fading. However, it has some limitations \cite{RNSM17} such as reduced spectral efficiency and high out-of-band (OOB) emissions. Therefore, new waveforms are proposed to address these issues. Even though new waveforms were not selected by the third generation partnership project (3GPP) for 5G new radio, these schemes remain as an interesting choice for future system development\cite{AJMR19}. Among various new waveforms, filter bank multi-carrier (FBMC) \cite{QDZW20,BFB11,FSTW14,TLLZ15} has been intensively studied. It consists of time-frequency well localized subcarrier filters, and the signal on each subcarrier is filtered individually. In fact, FBMC has been recognized as a competitive candidate in some new scenarios such as carrier aggregation, cognitive radio with spectrum sensing, and even in dense wavelength multiplexing passive optical network (DWDM-PON) based fronthaul using a radio-over-fiber technique \cite{SJSJ19}. The most popular form of FBMC is offset quadrature amplitude modulation (FBMC/OQAM) \cite{BFB11,PCNL02}, which splits the input QAM symbols into real and imaginary parts and requires orthogonality between subcarriers in the real domain for improving the time-frequency localization of the waveform and achieving maximal spectral efficiency.



\subsection{Prior Work}
It is well known that the orthogonality of FBMC/OQAM would be destroyed by the frequency-selective effect of channels \cite{TTMM07}, resulting in inter-symbol interference (ISI) and inter-carrier interference (ICI). To solve this issue, equalization should be implemented at the receiver to restore the orthogonality. Conventionally, it is assumed that the subcarrier filter length is relatively large compared to the channel length. With this assumption, a flat fading channel model can be assumed at each subcarrier \cite{CJRA08}. Thus, the single-tap equalization borrowed from OFDM can be directly implemented in FBMC/OQAM, which relies on linear receivers such as maximum ratio combining (MRC), zero forcing (ZF) and minimum mean squared error (MMSE). However, in highly frequency-selective channels, the above assumption is not accurate, and the performance of the single-tap equalization will be degraded. To improve the equalization accuracy, multi-tap equalization was raised in \cite{TTMM07,DSLJ08}, and the same idea was extended to the MIMO case \cite{AIJL09,TAJM11,MCAI12}. In \cite{AANL15}, frequency spreading equalization (FSE) was applied for FBMC-based MIMO systems at the expense of higher computational complexity, which can be regarded as a kind of multi-tap equalization per subcarrier. For other equalization schemes, \cite{XMDG16} proposed an architecture of transmission and reception to approximate the ideal frequency-selective precoder and linear receiver. The architecture was implemented by linearly combining conventional MIMO linear transceivers, which are applied to sequential derivatives of the original filter bank.

In recent years, researchers find that in massive MIMO, the interference brought by the channels, including ISI, ICI and inter-user interference (IUI), are expected to decrease as the number of BS antennas increases. It proves to be valid even in the case of strong channel selectivity and for conventional single-tap subcarrier equalizers \cite{ANLB14}, which results in low implementation complexity. The so-called ``self-equalization'' property was further analyzed in \cite{ANLD17}, where the authors demonstrated that the signal-to-interference-plus-noise ratio (SINR) performance of FBMC-based massive MIMO system would be limited with infinite number of BS antennas. This is due to the residual interference caused by the correlation between the equalizer taps and the channel impulse responses. Based on this, the same authors proposed efficient power delay profile (PDP) equalizers in \cite{AFFB18} to remove the above limitation, and arbitrarily large SINR values can be achieved. Besides, \cite{FXFJ18} theoretically characterized the mean squared error (MSE) of the estimated symbols based on linear receivers, and provided some insights into the self-equalization property from the perspective of MSE.

\subsection{Motivations}
However, the computational complexity of the scheme in \cite{AIJL09} is proportional to the cube of the number of BS antennas \cite{TAJM11}, and several sets of parallel analysis filter banks (AFB) should be equipped at each BS antenna for the scheme in \cite{XMDG16}. In the case of massive MIMO, the complexity of these two schemes may be unaffordable. The scheme in \cite{TAJM11} designed multi-tap equalizers directly at the symbol rate, and potential loss of channel information caused by FBMC demodulation was not considered, which degrades the system performance. Besides, the limitation analyzed in \cite{ANLD17} only holds when the number of BS antennas grows to infinity. In practice, the number of BS antennas is finite and may not be sufficiently large to completely characterize the interference as PDPs. Thus, the PDP equalizers may perform not well. Moreover, the scheme in \cite{AFFB18} requires not only instantaneous but statistical channel state information (CSI), which increases the complexity of channel estimation. The above findings motivate us to find a more general and effective equalization scheme to overcome the shortcomings of the prior works.

\subsection{Main Contributions}
In this paper, we propose a two-stage design scheme of equalizers with low implementation complexity, which can remove the limitation of the SINR performance effectively while only requiring the instantaneous CSI. Meanwhile, even with finite number of BS antennas, the interference brought by the channels can be almost eliminated. For clarity, we summarize our contributions as follows:
\begin{itemize}
\item We propose a two-stage design scheme of equalizers for uplink multi-user FBMC/OQAM-based massive MIMO system in highly frequency-selective channels. The general idea is that in Stage-1, we design high-rate equalizers working before FBMC demodulation to avoid the potential loss of channel information obtained at the BS. In Stage-2, we transform the high-rate equalizers into the low-rate equalizers after FBMC demodulation to reduce the implementation complexity. The scheme overcomes the shortcomings of the mentioned equalization schemes while performing well.

\item We develop a novel approach to transform the high-rate equalizers into low-rate equalizers in Stage-2, which combines two-step decimation with least squares (LS) method. We find the efficacy of the existing approaches is not ideal and analyze the reasons for it. The proposed approach addresses these issues effectively, and the low-rate equalizers have almost the same performance with the high-rate equalizers.

\item We derive the theoretical expression of SINR performance of the proposed scheme, which shows that the interference caused by the frequency-selective channels can be almost eliminated even with finite number of BS antennas. In this case, by increasing the transmit signal-to-noise ratio (SNR), the SINR value can approach the signal-to-interference ratio (SIR) upper bound, i.e., the SIR in single-input single-output (SISO) scenario when the channel response is unit impulse. All the above analyses are evaluated and confirmed through numerical simulations.

\end{itemize}

The rest of this paper is organized as follows. Section II presents the system model of FBMC/OQAM-based massive MIMO system. Section III proposes a two-stage design scheme of equalizers. Section IV analyzes the performance of the scheme. Section V provides simulation results. Section VI concludes this paper.

\emph{Notation}: Matrices are denoted by bold uppercase letters (e.g., $\mathbf{X}$). Vectors are represented by bold \emph{italic} letters (e.g., $\boldsymbol{x}$). Scalars are denoted by normal font letters (e.g., $x$). $j\triangleq\sqrt{-1}$. The real and complex number fields are denoted by $\mathbb{R}$ and $\mathbb{C}$, respectively. $\lceil x \rceil$ represents the minimum integer no smaller than $x$. $\Re\{x\}$ and $\Im\{x\}$ represent real and imaginary part of a complex number $x$. $x^*$ and $\vert x \vert$ stand for the conjugation and absolute value of $x$, respectively. The distribution of a circularly symmetric complex Gaussian random variable $x$ with mean $\mu$ and variance $\sigma^2$ is denoted by $x\sim\mathcal{CN}(\mu, \sigma^2)$. The linear convolution is denoted by $\star$. $\mathbb{E}[x]$ represents the statistical expectation of $x$. $[\boldsymbol{x}]_a$ and $\Vert\boldsymbol{x}\Vert$ are the $a$-th entry and the two-norm of $\boldsymbol{x}$, respectively. $[\mathbf{X}]_{a,b}$ and $\text{tr}(\mathbf{X})$ are the entry in the $a$-th row and $b$-th column and trace of $\mathbf{X}$, respectively. $(\cdot)^{\rm T}$, $(\cdot)^{\rm H}$ stand for transpose and Hermitian transpose. $\mathbf{I}_a$ denotes an identity matrix with size $a\times a$. $\mathbf{0}$ denotes an all-zero matrix with appropriate dimension.

\section{System Model}

\begin{figure*}
\centering
\includegraphics[width=6in]{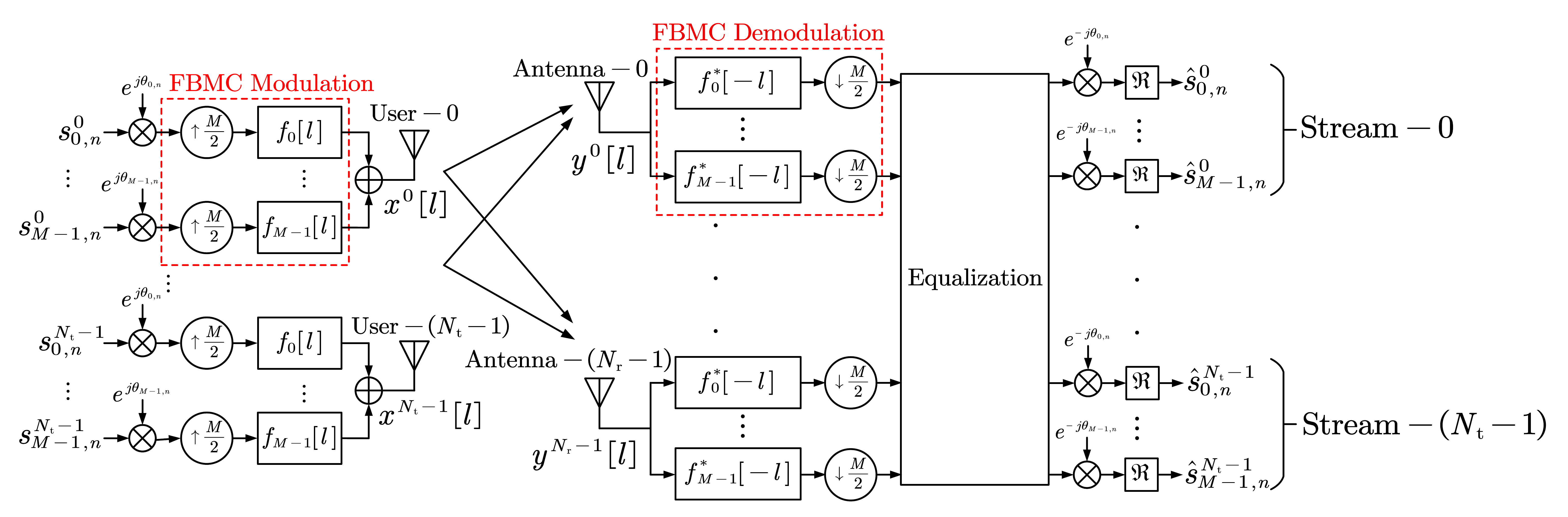}
\centering
\caption{Block diagram of uplink FBMC/OQAM-based massive MIMO system.}
\label{FBMC-MIMO}
\end{figure*}

In this section, we present the uplink model of a single-cell multi-user FBMC/OQAM-based massive MIMO system, which is shown in Fig. \ref{FBMC-MIMO}. Assume that there are $N_{\rm t}$ single-antenna users communicating with an $N_{\rm r}$-antenna BS using the same time-frequency resources. Besides, the total number of subcarriers is assumed to be $M$. For the $u$-th user, the input OQAM symbols $s_{m,n}^u, \forall m,n$ are multiplied by $e^{j\theta_{m,n}}$ to adjust the phase, where $\theta_{m,n}=\frac{\pi}{2}(m+n)$. Accordingly, each symbol has a $\pm\frac{\pi}{2}$ phase difference with its adjacent neighbors in both time and frequency, which can avoid interference between the symbols. Then, the symbols are processed by FBMC modulator which generates the signal $x^u[l]$ that is transmitted by the antenna, where $m$ and $n$ denote subcarrier index and symbol instant, respectively. We assume that the OQAM symbols are independent over $m$, $n$, $u$ and identically distributed (i.i.d) with zero mean and power $P_{\rm s}$ such that $\mathbb{E} \left\{s_{m,n}^u(s_{m,n}^u)^*\right\} = P_{\rm s}$. The OQAM symbols are generated by extracting the real and imaginary parts of the complex QAM symbols according to the rules described in \cite{VISR09}. The duration of each QAM symbol is represented as $T$, while $\frac{T}{2}$ is the duration of each OQAM symbol. In FBMC modulator, the symbol after phase adjustment is upsampled by $\frac{M}{2}$, filtered by the synthesis filter, and then aggregated with the signals at other subcarriers. The synthesis filter at the $m$-th subcarrier is expressed as $f_m[l] = p[l]e^{j\frac{2\pi ml}{M}}$, where $l$ denotes the sample index corresponding to the sampling rate $\frac{M}{T}$. The symmetric real-valued pulse $p[l]$ represents the impulse response of the prototype filter of the FBMC system, which has been designed such that $p[l]\star p^*[-l]$ is a Nyquist pulse. The length of the prototype filter is expressed as $L_{\rm f}=\kappa M$, where $\kappa$ is called the overlapping factor. Thus, the equivalent discrete-time baseband FBMC transmit signal $x^u[l]$ of the $u$-th user can be expressed as
\begin{align}\label{xul}
x^u[l]=\sum_{n=-\infty}^{+\infty}\sum_{m=0}^{M-1} s_{m,n}^u F_{m,n}[l],
\end{align}
where $F_{m,n}[l] \triangleq f_m[l-\frac{nM}{2}] e^{j\theta_{m,n}}$. Specifically, $F_{m,n}[l], \forall m,n$ are orthogonal in the real domain, i.e., $\Re\left\{\sum_{l=-\infty}^{+\infty} F_{m,n}[l] F_{m',n'}^*[l] \right\} = \delta_{mm'}\delta_{nn'}$, where $\delta_{mm'}$ is the Kronecker delta with $\delta_{mm'}=1$ if $m=m'$ and zero otherwise.

Before received by the $r$-th BS antenna, $x^u[l]$ is transmitted through a frequency-selective Rayleigh fading channel $h^{r,u}[l]$ and corrupted by an additive white Gaussian noise (AWGN) $z^r[l]$. We assume that the channel length $L_{\rm h}$ is less than $M$. Besides, for the $u$-th user, the channel responses $h^{r,u}[l] \sim \mathcal{CN}(0, q^u[l])$ for $\forall r, l$, and different taps are assumed to be independent, where $q^u[l]$ represents the channel PDP. We also assume that the channels corresponding to different users and different BS antennas are independent, and the channels are perfectly known at the BS. For $\forall l$, the noise $z^r[l] \sim \mathcal{CN}(0, \sigma_{\rm z}^2)$ and is assumed to be independent over $l$ and $r$.

The aggregate received signal at the $r$-th BS antenna from all users can be expressed as
\begin{align}\label{yrl}
y^r[l] = \sum_{u=0}^{N_{\rm t}-1} \left( x^u[l] \star h^{r,u}[l] \right)+z^r[l].
\end{align}
It is processed by FBMC demodulator and fed to equalization module together with the received signals at other BS antennas. The equalization module can alleviate the impact of frequency-selectivity of the channels and recover $N_{\rm t}$-stream signals of the $N_{\rm t}$ users. The estimate of $s_{m,n}^u$, denoted by $\hat{s}_{m,n}^u$, can be obtained by  compensating the phase and then taking the real part of the signal at the $m$-th subcarrier of the $u$-th stream at the output of the equalization module. FBMC demodulator consists of analysis filters $f_m^*[-l], m=0,\cdots,M-1$ and $\frac{M}{2}$-fold decimators.

\section{Two-stage Design Scheme of Equalizers}\label{MCE}
In order to overcome the shortcomings of the prior works and achieve better equalization performance, our goal is to design multi-tap equalizers working at the symbol rate (i.e., at the OQAM symbol rate of $\frac{2}{T}$ after $\frac{M}{2}$-fold decimation). Hereafter, we call the equalizers working at the symbol rate as low-rate equalizers.

\subsection{General Idea}
Unlike the prior work in \cite{TAJM11}, which designed multi-tap equalizers directly at the symbol rate, we first design equalizers working at a higher symbol rate, then transform the high-rate equalizers into low-rate equalizers. The reason is that the low-rate equalizers equalize the channels based on $\mathcal{D}(y^r[l])$, where $\mathcal{D}(\cdot)$ denotes the operation of FBMC demodulation. According to the data processing inequality, any operation on $y^r[l]$ will result in potential loss of the channel information obtained at the BS, unless it is some sufficient statistics. However, there is no proof that $\mathcal{D}(\cdot)$ has such property. Therefore, the channel information seen by the equalizers is not the exact CSI. To address this issue, we first design the high-rate multi-tap equalizers working before FBMC demodulation.

\subsection{High-rate Equalizers}\label{HE}
Generally, many equalizer design methods can be applied to the first stage. Here, we take the most common linear equalization with ZF and MMSE criteria as an example. Collecting the received signal on all BS antennas, (\ref{yrl}) can be rewritten as
\begin{align}
\boldsymbol{y}[l]=\mathbf{H}[l] \star \boldsymbol{x}[l] + \boldsymbol{z}[l],
\end{align}
where $\boldsymbol{y}[l] \triangleq \left[y^0[l],\cdots,y^{N_{\rm r}-1}[l]\right]^{\rm T}$, $\boldsymbol{x}[l] \triangleq \left[x^0[l],\cdots,x^{N_{\rm t}-1}[l]\right]^{\rm T}$, $\boldsymbol{z}[l] \triangleq \left[z^0[l],\cdots,z^{N_{\rm r}-1}[l]\right]^{\rm T}$. $\mathbf{H}[l]$ is an $N_{\rm r}\times N_{\rm t}$ matrix with the entry in the $r$-th row and $u$-th column given by $h^{r,u}[l]$. Besides, the convolution of $\mathbf{H}[l]$ and $\boldsymbol{x}[l]$ is defined as $\mathbf{H}[l] \star \boldsymbol{x}[l] \triangleq \sum_{\ell=-\infty}^{+\infty} \mathbf{H}[l-\ell] \boldsymbol{x}[\ell]$. Let $\mathbf{G}[l]\in\mathbb{C}^{N_{\rm t}\times N_{\rm r}}$ denote the high-rate equalizer matrix with the number of taps $L_{\rm g}$. $g^{r,u}[l] \triangleq \left[ \mathbf{G}[l] \right]_{r,u}$ is the impulse response of the equalizer to recover the $u$-th user signal from the $r$-th BS antenna. Consequently, the received signal after equalization can be expressed as
\begin{align}\label{equalization}
\boldsymbol{\hat{x}}[l]
=\mathbf{G}[l] \star \boldsymbol{y}[l]
=\mathbf{G}[l] \star \mathbf{H}[l] \star \boldsymbol{x}[l] + \mathbf{G}[l] \star \boldsymbol{z}[l].
\end{align}
We will first derive the expression of $\mathbf{G}[l]$ in frequency domain, which is denoted by $\mathbf{\tilde{G}}(\omega) \triangleq \sum_{l=-\infty}^{+\infty} \mathbf{G}[l] e^{-j\omega l}$. Then, we use the frequency sampling (FS) method to obtain $\mathbf{G}[l]$.

\subsubsection{ZF Criterion}
In order to eliminate the interference of the received signal, it should satisfy $\mathbf{G}[l]\star\mathbf{H}[l]=\mathbf{\Delta}\left[l-\frac{\alpha M}{2}\right]$, where $\mathbf{\Delta}[l]=\mathbf{I}_{N_{\rm t}}$ if $l=0$ and $\mathbf{0}$ otherwise. $\frac{\alpha M}{2}$ is a delay term to optimize the equalizer performance, where $\alpha$ is an integer with $0 \leq \alpha \leq \lceil\frac{2(L_{\rm h}+L_{\rm g}-1)}{M}\rceil-1$, and $\frac{M}{2}$ corresponds to the factor of the subsequent decimation at the receiver. In frequency domain, $\mathbf{\tilde{G}}(\omega)$ can be calculated by
\begin{align}
\mathbf{\tilde{G}}(\omega) = \left( \mathbf{\tilde{H}}^{\rm H}(\omega) \mathbf{\tilde{H}}(\omega) \right)^{-1} \mathbf{\tilde{H}}^{\rm H}(\omega) e^{-j\omega\frac{\alpha M}{2}}.
\end{align}

\subsubsection{MMSE Criterion}
To eliminate the impact of noise, $\mathbf{G}[l]$ is selected to minimize the estimation error $\mathbb{E}\left\{\big\Vert \boldsymbol{\hat{x}}[l] - \boldsymbol{x}\left[l-\frac{\alpha M}{2}\right] \big\Vert^2\right\}$. In frequency domain, $\mathbf{\tilde{G}}(\omega)$ can be calculated by
\begin{align}
\mathbf{\tilde{G}}(\omega)
=\left( \mathbf{\tilde{H}}^{\rm H}(\omega) \mathbf{\tilde{H}}(\omega) + \frac{\sigma_{\rm z}^2}{P_{\rm s}} \mathbf{I}_{N_{\rm t}} \right)^{-1} \mathbf{\tilde{H}}^{\rm H}(\omega) e^{-j\omega\frac{\alpha M}{2}}.
\end{align}

Then, we use the FS method to obtain the equalizers in time domain with finite length. Specifically, $\mathbf{\tilde{G}}(\omega)$ is sampled with equal intervals, i.e., $\mathbf{\tilde{G}}(k) = \mathbf{\tilde{G}}(\frac{2\pi k}{L_{\rm g}}), k=0,\cdots,L_{\rm g}-1$. Thus, $\mathbf{G}[l]$ can be obtained by the inverse discrete Fourier transform (IDFT) of $\mathbf{\tilde{G}}(k)$, i.e.,
\begin{align}
\mathbf{G}[l] = \frac{1}{L_{\rm g}} \sum_{k=0}^{L_{\rm g}-1} \mathbf{\tilde{G}}(k) e^{j\frac{2\pi lk}{L_{\rm g}}}, \quad l=0,\cdots,L_{\rm g}-1.
\end{align}

It is well known that the power of the transmit signal is concentrated on the frequency bins $\frac{2\pi m}{M}, m=0,\cdots,M-1$, which indicates that these frequency bins of the channel responses should be perfectly equalized. Therefore, the minimum number of sampling points (in other words, the equalizer length $L_{\rm g}$) is $M$.

\subsection{From High-rate to Low-rate Equalizers}
In this subsection, we aim to transforming the above high-rate equalizers into low-rate equalizers, since the high-rate equalizers require high quality electronic devices, which results in high power-consumption. Although there are some existing transform methods, their efficacy is not ideal. We first analyze the reasons for the unsatisfactory efficacy. Then, we propose a novel approach and explain how it addresses these issues.

\subsubsection{Existing Methods}
For subcarrier-$m$, $g^{r,u}[l]$ can be replaced by any other high-rate equalizer as long as it has the same frequency response with $g^{r,u}[l]$ in the frequency range $[\frac{2\pi(m-1)}{M},\frac{2\pi(m+1)}{M})$ which corresponds to the pass band of $f_m^*[-l]$. We denote the equalizer that satisfies this condition by $\check{g}_m^{r,u}[l]$. As far as we know, there are mainly two methods to construct $\check{g}_m^{r,u}[l]$, based on which we can obtain the same low-rate equalizer $\bar{g}_m^{r,u}[n]$. Fig. \ref{Method1&2Model} shows the whole evolving process of the receiver from high-rate to low-rate equalizer. The following two paragraphs will illustrate the two methods in detail, respectively. We denote these two kinds of $\check{g}_m^{r,u}[l]$ by $\check{g}_{{\rm (1)},m}^{r,u}[l]$ and $\check{g}_{{\rm (2)},m}^{r,u}[l]$ to distinguish them.

\begin{figure}
\centering
\includegraphics[width=3.2in]{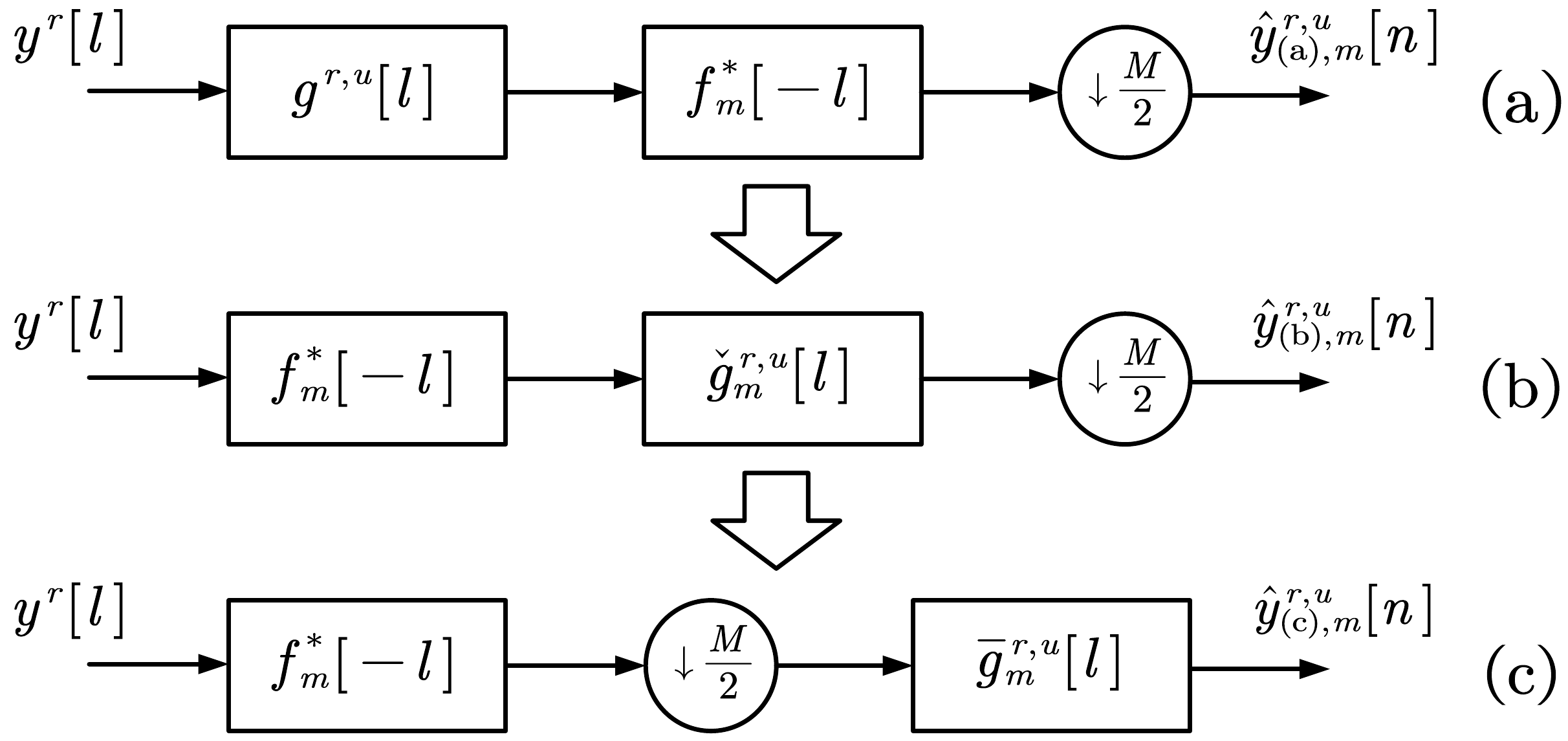}
\centering
\caption{Evolvement of the receiver from channel equalization to subcarrier equalization at subcarrier-$m$.}
\label{Method1&2Model}
\end{figure}

\begin{figure*}[htp]
\centering
\includegraphics[width=5in]{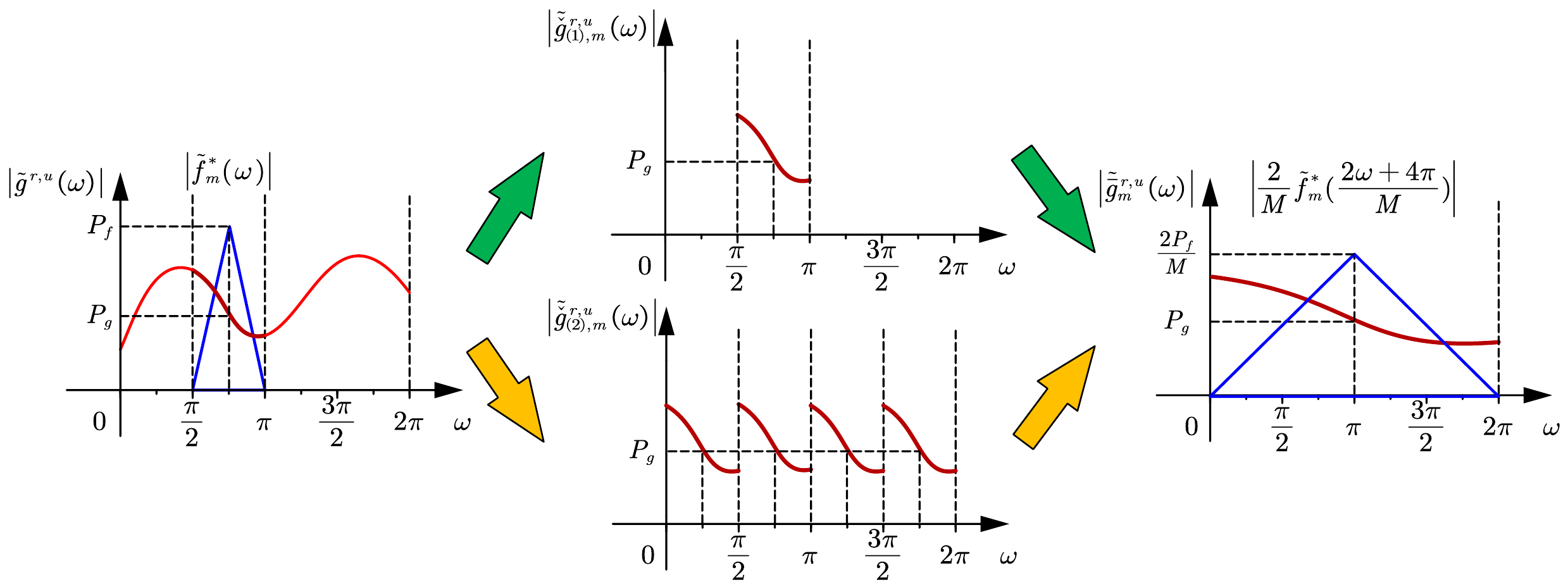}
\centering
\caption{Schematic drawings of equalizers in frequency domain in each evolving stage, where \textbf{Method 1} evolves along upper arrows and \textbf{Method 2} evolves along lower arrows. The parameters are set as $M=8$ and $m=3$.}
\label{Method1&2FR}
\end{figure*}

\textbf{Method 1}: The first method is shown in \cite{AFFB18}, where the authors let the frequency response of $\check{g}_m^{r,u}[l]$ outside the range $[\frac{2\pi(m-1)}{M},\frac{2\pi(m+1)}{M})$ be zero, i.e., $\check{g}_{{\rm (1)},m}^{r,u}[l]=g^{r,u}[l]\star\left(\frac{2}{M}\text{sinc}\left( \frac{2l}{M} \right) e^{j\frac{2\pi ml}{M}}\right)$. Here, $\text{sinc}(t)\triangleq\frac{\sin(\pi t)}{\pi t}$ and $\frac{2}{M}\text{sinc}\left( \frac{2l}{M} \right) e^{j\frac{2\pi ml}{M}}$ acts as a band-pass filter with pass band $[\frac{2\pi(m-1)}{M},\frac{2\pi(m+1)}{M})$. Therefore, $\bar{g}_m^{r,u}[n]$ could be obtained by decimating $\check{g}_{{\rm (1)},m}^{r,u}[l]$ with factor $\frac{M}{2}$ and then multiplying by $\frac{M}{2}$, i.e., $\bar{g}_m^{r,u}[n] = \frac{M}{2}\left(\check{g}_{{\rm (1)},m}^{r,u}[l]\right)_{\downarrow\frac{M}{2}}$.

\textbf{Method 2}: The second method is shown in \cite{DJZZ09}, where the authors proposed that the frequency response of $\check{g}_{{\rm (2)},m}^{r,u}[l]$ was periodic with a smallest period $\frac{4\pi}{M}$. It means that the DTFT of $\check{g}_{{\rm (2)},m}^{r,u}[l]$, which is denoted by $\tilde{\check{g}}_{{\rm (2)},m}^{r,u}(\omega)$, can be formed by repeating $\tilde{g}^{r,u}(\omega)$ in the range of $[\frac{2\pi(m-1)}{M},\frac{2\pi(m+1)}{M})$ across the entire frequency range $[0,2\pi)$ with a period $\frac{4\pi}{M}$. Therefore, $\bar{g}_m^{r,u}[n]$ can be obtained from the $\frac{M}{2}$-fold decimated version of $\check{g}_{{\rm (2)},m}^{r,u}[l]$, i.e., $\bar{g}_m^{r,u}[n] = \left(\check{g}_{{\rm (2)},m}^{r,u}[l]\right)_{\downarrow\frac{M}{2}}$. For clarity, in Fig. \ref{Method1&2FR}, we give the schematic drawings of each evolving stage in frequency domain.

\begin{proposition}\label{(a)=(c)}
Through the above two methods, the same low-rate equalizer $\bar{g}_m^{r,u}[n]$ can be obtained. Besides, the receivers shown in Fig. \ref{Method1&2Model}-(a) and Fig. \ref{Method1&2Model}-(c) are equivalent from the perspective of frequency domain.
\end{proposition}
\begin{proof}
Please refer to Appendix \ref{Proof(a)=(c)}.
\end{proof}

\subsubsection{Reasons for the Unsatisfactory Efficacy}\label{SE}
However, we find that the performance of $\bar{g}_m^{r,u}[n]$ obtained by the aforementioned methods is not satisfying. The possible reasons are as follows:

\textbf{R1}: In practical applications, no matter how well the prototype filter is designed, its frequency response always has some OOB emissions, though they seem fairly tiny;

\textbf{R2}: The band-pass filter is not ideal, and there is a deviation from the desired filtering results;

Due to \textbf{R1} and \textbf{R2}, $f_m^*[-l]$ and $\tilde{\check{g}}_m^{r,u}(\omega)$ are actually not strictly band-limited in the range $[\frac{2\pi(m-1)}{M},\frac{2\pi(m+1)}{M})$. Thus, (\ref{ineqi}) in Appendix \ref{Proof(a)=(c)} does not hold any more.

\textbf{R3}: There exists a response jump with a relatively high probability, which can be found explicitly in Fig. \ref{Method1&2FR}. To be more specific, for our desired subcarrier equalizer response $\tilde{\bar{g}}_m^{r,u}(\omega)$, $\tilde{\bar{g}}_m^{r,u}(0)$ is not necessarily equal to $\tilde{\bar{g}}_m^{r,u}(2\pi)$. Therefore, when transforming $\tilde{\bar{g}}_m^{r,u}(\omega)$ into $\bar{g}_m^{r,u}[n]$, the approximation error will be relatively large in the jump area.

\subsubsection{Proposed Approach}
For \textbf{R1}, it can be considered that $\tilde{F}_m^*(\omega)$ is band-limited in a wider range than $[\frac{2\pi(m-1)}{M},\frac{2\pi(m+1)}{M})$ to take the main emissions into the range, thus avoiding the loss of the information of $g^{r,u}[l]$. From this perspective, we can extend \textbf{Method 1} and give the following proposition.
\begin{proposition}
If $\tilde{F}_m^*(\omega)$ is limited in $[\frac{\pi(m-1)}{D_1},\frac{\pi(m+1)}{D_1})$ in frequency domain, the cascade of $g^{r,u}[l]$, $f_m^*[-l]$ and $D_1$-fold decimator is equivalent to the cascade of $f_m^*[-l]$, $D_1$-fold decimator and $\bar{g}_m^{r,u}[n]$. Here, $D_1=\frac{M}{2^\eta}$, and $\eta$ is an integer ranging from $1$ to $\log_2M$. $\bar{g}_m^{r,u}[n]$ is obtained by decimating $\check{g}_m^{r,u}[l]$ with the factor $D_1$ and then multiplying by $D_1$. $\check{g}_m^{r,u}[l]$ is obtained by passing $g^{r,u}[l]$ through an ideal band-pass filter with pass-band $[\frac{\pi(m-1)}{D_1},\frac{\pi(m+1)}{D_1})$.
\end{proposition}

\begin{figure}[b]
\centering
\includegraphics[width=3.5in]{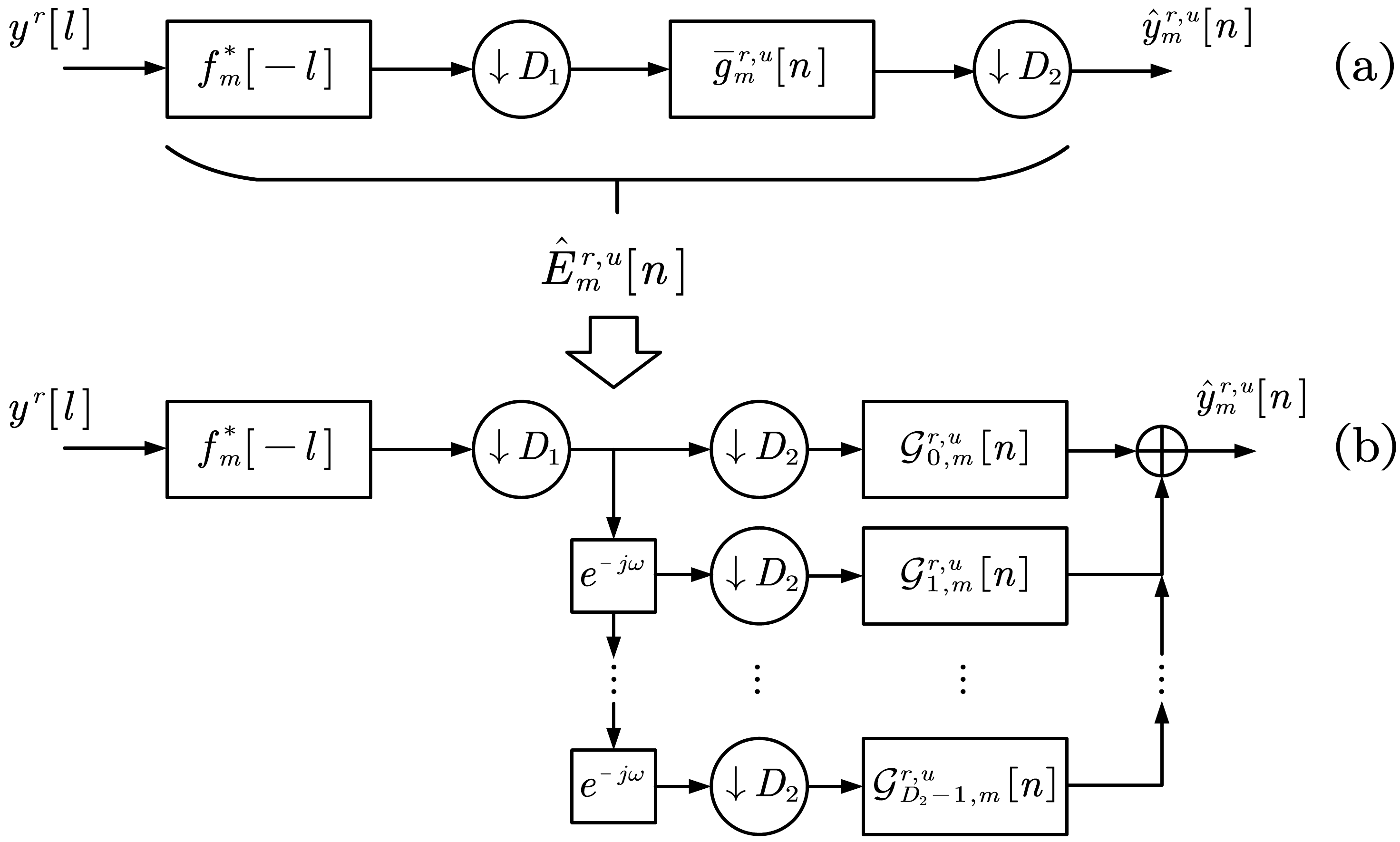}
\centering
\caption{Block diagram of the two-step decimation.}
\label{TwostepDecimation}
\end{figure}

The proof is similar with \emph{Part 1} in Appendix \ref{Proof(a)=(c)}. In this way, $\bar{g}_m^{r,u}[n]$ works at a higher symbol rate. For example, when $D_1=\frac{M}{4}$, the equalizer works at twice symbol rate; when $D_1=\frac{M}{8}$, it works at quadruple symbol rate. Therefore, a $D_2$-fold decimator should be added after $\bar{g}_m^{r,u}[n]$ to recover the symbol rate, and $D_2$ should satisfy $D_1D_2=\frac{M}{2}$. The receiver corresponding to this so called two-step decimation is shown in Fig. \ref{TwostepDecimation}-(a), where the equivalent channel between $y^r[l]$ and $\hat{y}_m^{r,u}[n]$ is expressed as $\hat{E}_m^{r,u}[n]= \left( \left( f_m^*[-l] \right)_{\downarrow D_1} \star \bar{g}_m^{r,u}[n] \right)_{\downarrow D_2}$.
To let the equalizers work at the symbol rate, a further implementation of the two-step decimation is shown in Fig. \ref{TwostepDecimation}-(b). It can be seen that $\bar{g}_m^{r,u}[n]$ is divided into $D_2$ equalizers $\mathcal{G}_{\ell,m}^{r,u}[n], \ell=0,\cdots,D_2-1$, and each equalizer is obtained by $\mathcal{G}_{\ell,m}^{r,u}[n] = \bar{g}_m^{r,u} [D_2n+\ell]$. The explanation is that the DTFT of $\bar{g}_m^{r,u}[n]$ can be rewritten as
\begin{align}
\tilde{\bar{g}}_m^{r,u}(\omega)
=\sum_{\ell=0}^{D_2-1} e^{-j\omega\ell} \tilde{\mathcal{G}}_{\ell,m}^{r,u} (D_2\omega),
\end{align}
where $\tilde{\mathcal{G}}_{\ell,m}^{r,u}(\omega) \triangleq \sum_{n=-\infty}^{+\infty} \bar{g}_m^{r,u} [D_2n+\ell] e^{-j\omega n}$. Moreover, according to noble identity in multi-rate filter bank theory \cite{PPV93}, the cascade of $\tilde{\mathcal{G}}_{\ell,m}^{r,u} (D_2\omega)$ and $D_2$-fold decimator can be altered by the cascade of $D_2$-fold decimator and $\tilde{\mathcal{G}}_{\ell,m}^{r,u} (\omega)$. Therefore, the receiver shown in Fig. \ref{TwostepDecimation}-(b) is equivalent to that in Fig. \ref{TwostepDecimation}-(a).

For \textbf{R2} and \textbf{R3}, we use the LS method to let the performance of $\bar{g}_m^{r,u}[n]$ approach that of $g^{r,u}[l]$. To be specific, $\tilde{\bar{g}}_m^{r,u}(\omega)$ is selected to minimize the estimation error, i.e.,
\begin{align}
\tilde{\bar{g}}_m^{r,u}(\omega) = \argmin \frac{1}{2\pi} \int_{0}^{2\pi} \left|\tilde{\hat{E}}_{{\rm (a)},m}^{r,u}(\omega) - \tilde{\hat{E}}_{{\rm (c)},m}^{r,u}(\omega) \right|^2 d\omega.
\end{align}
According to Parseval's theorem, it should satisfy
\begin{align}\label{bargmruargmin}
\bar{g}_m^{r,u}[n] = \argmin \sum_{n=-\infty}^{+\infty} \left|\hat{E}_{{\rm (a)},m}^{r,u}[n] - \hat{E}_{{\rm (c)},m}^{r,u}[n] \right|^2,
\end{align}
where $\hat{E}_{{\rm (a)},m}^{r,u}[n] = \left( g^{r,u}[l] \star f_m^*[-l] \right)_{\downarrow D_1}$, $\hat{E}_{{\rm (c)},m}^{r,u}[n] = \left( f_m^*[-l] \right)_{\downarrow D_1} \star \bar{g}_m^{r,u}[n]$. $\tilde{\hat{E}}_{{\rm (a)},m}^{r,u}(\omega)$ and $\tilde{\hat{E}}_{{\rm (c)},m}^{r,u}(\omega)$ are the DTFT of $\hat{E}_{{\rm (a)},m}^{r,u}[n]$ and $\hat{E}_{{\rm (c)},m}^{r,u}[n]$, respectively.
Assume that the length of $\bar{g}_m^{r,u}[n]$ is $L'_{\rm g}$. We denote $\mathbf{F}_m^{\rm cir}$ as a circular matrix with $\frac{L_{\rm f}}{D_1} + L'_{\rm g}-1$ rows and $L'_{\rm g}$ columns. Besides, its first column is given by zero-padded version of $\left(f_m^*[-l]\right)_{\downarrow D_1}$

$\boldsymbol{\bar{g}}_m^{r,u}$ and $\boldsymbol{\hat{e}}_m^{r,u}$ are column vectors with size $L'_{\rm g}$ and $\frac{L_{\rm f}}{D_1} + L'_{\rm g}-1$, respectively. $\boldsymbol{\bar{g}}_m^{r,u}$ is composed by $\bar{g}_m^{r,u}[n]$. If $L'_{\rm g} < \frac{L_{\rm g}}{D_1}$, $\boldsymbol{\hat{e}}_m^{r,u}$ is composed by the first $ \frac{L_{\rm f}}{D_1} + L'_{\rm g}-1$ entries of $\hat{E}_{{\rm (a)},m}^{r,u}[n]$, otherwise, by zero-padded version of $\hat{E}_{{\rm (a)},m}^{r,u}[n]$. In this way, (\ref{bargmruargmin}) can be further transformed into
\begin{align}\label{LS}
\boldsymbol{\bar{g}}_m^{r,u}
=\argmin \left\Vert \mathbf{F}_m^{\rm cir} \boldsymbol{\bar{g}}_m^{r,u} - \boldsymbol{\hat{e}}_m^{r,u} \right\Vert^2
=\mathbf{\bar{F}}_m \boldsymbol{\hat{e}}_m^{r,u},
\end{align}
where $\mathbf{\bar{F}}_m \triangleq \left( \left(\mathbf{F}_m^{\rm cir}\right)^{\rm H} \mathbf{F}_m^{\rm cir} \right)^{-1} \left(\mathbf{F}_m^{\rm cir}\right)^{\rm H}$. Notice that if the analysis filters have been designed, $\mathbf{\bar{F}}_m^{\rm cir}$ will be fixed. Thus, the computational complexity mainly comes from the calculation of $\boldsymbol{\hat{e}}_m^{r,u}$ and multiplication between the fixed matrix and $\boldsymbol{\hat{e}}_m^{r,u}$. By combining the above two-step decimation and the LS method, we can maintain the SINR performance while reducing the implementation complexity.

\section{Performance Analysis}\label{chara}
In this section, we analyze the SINR performance of the proposed scheme based on the high-rate equalizers. Although the final low-rate equalizers are not completely equivalent to the high-rate equalizers, we will show that the performance gap between them is negligible in the subsequent simulation results. Besides, the complexity of calculating the equalizer coefficients are also given in this section.

\subsection{SNIR Performance Analysis}\label{SP}
According to the law of large numbers, it can be easily derived that when $\alpha=1$ and $M\geq2(L_{\rm h}-1)$, as the number of BS antennas tends to infinity, $\mathbf{G}[l]\star\mathbf{H}[l]$ and $P_{{\rm z},m,n}^u$ converge almost surely to $\mathbf{\Delta}\left[l-\frac{\alpha M}{2}\right]$ and $0$, respectively. Here, $P_{{\rm z},m,n}^u$ denotes the power of the noise $z_{m,n}^u$ contained in the final estimated data symbol $\hat{s}_{m,n}^u$. It indicates that the impacts of the channels and noise can be completely eliminated with infinite number of BS antennas, and arbitrary SINR can be achieved theoretically.

However, in practice, the number of BS antennas is not large enough to be regarded as infinity. When $N_{\rm r}$ is finite, these impacts may not be completely eliminated. In this case, we use a variety of random variables to characterize the equalization error, which can be expressed by
\begin{align}
\mathbf{G}[l] \star \mathbf{H}[l] = \mathbf{\Delta}\left[l-\frac{\alpha M}{2}\right] + \mathbf{\Psi}[l],
\end{align}
where $\mathbf{\Psi}[l]$ is an $N_{\rm t}\times N_{\rm t}$ complex random matrix which represents the equalization error. Based on this, we will derive the expression of SINR to reveal the relation between $\mathbf{\Psi}[l]$ and SINR. By introducing the multi-tap equalizers, the demodulated signal $\boldsymbol{d}_{m,n}$ before taking the real part can be written as
\begin{align}\label{dmnnew}
\boldsymbol{d}_{m,n}=\sum_{n'=-\infty}^{+\infty}\sum_{m'=0}^{M-1} \mathbf{E}_{mm',nn'}\boldsymbol{s}_{m',n'}+\boldsymbol{z}_{m,n},
\end{align}
where
\begin{align}
\nonumber
\mathbf{E}_{mm',nn'}
\triangleq&\mathbf{E}_{mm'}[n-(n'-\alpha)] e^{j(\theta_{m',n'}-\theta_{m,n})},\\
\nonumber
\mathbf{E}_{mm'}[n]
\triangleq&\left( \mathcal{F}_{mm'}[l] \mathbf{I}_{N_{\rm t}} \star \mathbf{G}[l] \star \mathbf{H}[l] \right)_{\downarrow\frac{M}{2}},\\
\nonumber
\mathcal{F}_{mm'}[l]
\triangleq&f_{m'}[l] \star f_m^*[-l],\\
\nonumber
\boldsymbol{s}_{m,n} \triangleq & [s_{m,n}^0,\cdots,s_{m,n}^{N_{\rm t}-1}]^{\rm T}.
\end{align}
The subscript $\downarrow D$ denotes the $D$-fold decimation. As shown in Fig. \ref{FBMC-MIMO}, the interference coefficients are given by the real part of the entries of $\mathbf{E}_{mm',nn'}$, which is denoted by $\mathbf{R}_{mm',nn'}$. Besides, let $\boldsymbol{\hat{z}}_{m,n}$ denote the real part of the entries of $\boldsymbol{z}_{m,n}$. Accordingly, the estimate of $\boldsymbol{s}_{m,n}$ can be formulated as
\begin{align}
\boldsymbol{\hat{s}}_{m,n}=\sum_{n'=-\infty}^{+\infty}\sum_{m'=0}^{M-1} \mathbf{R}_{mm',nn'}\boldsymbol{s}_{m',n'}+\boldsymbol{\hat{z}}_{m,n}.
\end{align}
The entry of $\mathbf{R}_{mm',nn'}$ in the $u$-th row and $u'$-th column can be further formulated as
\begin{align}
\nonumber
R_{mm',nn'}^{uu'} = &\left( \boldsymbol{\check{\psi}}^{uu'} \right)^{\rm T} \boldsymbol{\check{\mathcal{F}}}_{mm',nn'}\\
&+\Re\left\{ \mathcal{F}_{mm',nn'}\left[\frac{\alpha M}{2}\right] \right\}\delta_{uu'},
\end{align}
where
\begin{align}
\nonumber
\boldsymbol{\check{\psi}}^{uu'}
&\triangleq \Big[ \Re\Big\{\psi^{uu'}[0],\cdots,\psi^{uu'}[M+L_{\rm h}-2] \Big\} ,\\
\nonumber
&\Im\Big\{\psi^{uu'}[0],\cdots,\psi^{uu'}[M+L_{\rm h}-2] \Big\} \Big]^{\rm T},\\
\nonumber
\boldsymbol{\check{\mathcal{F}}}_{mm',nn'}
&\triangleq \Big[ \Re\Big\{\mathcal{F}_{mm',nn'}[0], \cdots, \mathcal{F}_{mm',nn'}[M+L_{\rm h}-2] \Big\} ,\\
\nonumber
&\mspace{-10mu}-\Im\Big\{\mathcal{\mathcal{F}}_{mm',nn'}[0], \cdots, \mathcal{F}_{mm',nn'}[M+L_{\rm h}-2] \Big\} \Big]^{\rm T},
\end{align}
$\mathcal{F}_{mm',nn'}[l] \triangleq \mathcal{F}_{mm'}\left[\frac{(n-n'+\alpha)M}{2}-l\right] e^{j(\theta_{m',n'}-\theta_{m,n})}$, $\psi^{uu'}[l] \triangleq \left[ \mathbf{\Psi}[l] \right]_{u,u'}$. Therefore, the instantaneous power corresponding to $R_{mm',nn'}^{uu'}$ can be calculated as
\begin{align}
\nonumber
&\mathrel{\phantom{=}}P_{mm',nn'}^{uu'}\\
\nonumber
&=\left( R_{mm',nn'}^{uu'} \right)^2 P_{\rm s}\\
\nonumber
&=2\Re\left\{ \mathcal{F}_{mm',nn'}\left[\frac{\alpha M}{2}\right] \right\} P_{\rm s} \delta_{uu'} \left( \boldsymbol{\check{\psi}}^{uu'} \right)^{\rm T} \boldsymbol{\check{\mathcal{F}}}_{mm',nn'}\\
\nonumber
&\mathrel{\phantom{=}}+\left( \boldsymbol{\check{\psi}}^{uu'} \right)^{\rm T} \boldsymbol{\check{\mathcal{F}}}_{mm',nn'} \left( \boldsymbol{\check{\mathcal{F}}}_{mm',nn'} \right)^{\rm T} \boldsymbol{\check{\psi}}^{uu'} P_{\rm s}\\
&\mathrel{\phantom{=}}+\left(\Re\left\{ \mathcal{F}_{mm',nn'}\left[\frac{\alpha M}{2}\right] \right\}\right)^2 P_{\rm s} \delta_{uu'}.
\end{align}
The average power, with averaging over different channel realizations, can be calculated as $\bar{P}_{mm',nn'}^{uu'}=\mathbb{E}\left\{ P_{mm',nn'}^{uu'} \right\}$. Accordingly, the SINR at the $m$-th subcarrier and $n$-th symbol instant of the $u$-th user is defined as (\ref{SINR}) shown at the top of the next page.
\begin{figure*}[!t]
\normalsize
\begin{align}\label{SINR}
\text{SINR}_{m,n}^{u} \triangleq \frac{\bar{P}_{mm,nn}^{uu}} {\mathop{\sum_{n'=-\infty}^{+\infty} \sum_{m'=0}^{M-1}}\limits_{(m',n')\neq(m,n)} \bar{P}_{mm',nn'}^{uu'} + \sum_{n'=-\infty}^{+\infty} \sum_{m'=0}^{M-1} \mathop{\sum_{u'=0}^{N_{\rm t}-1}}\limits_{u'\neq u} \bar{P}_{mm',nn'}^{uu'} +\bar{P}_{{\rm z},m,n}^u}
\end{align}
\hrulefill
\vspace*{4pt}
\end{figure*}
From the equation, we can see that the interference coefficients depend on the values of $\bar{P}_{mm',nn'}^{uu'}$, which contain the statistics of $\psi^{uu'}[l]$. Taking the equalizers designed under ZF criterion as an example, in Appendix B in supplementary file, we roughly calculate the statistics of $\psi^{uu'}[l]$ due to the accurate values are hard to obtain. The conclusions are also applicable to the case of MMSE criterion. Thus, the average power $\bar{P}_{mm',nn'}^{uu'}$ can be calculated as
\begin{align}\label{barP}
\nonumber
\bar{P}_{mm',nn'}^{uu'}=
&\text{tr} \left\{ \boldsymbol{\check{\mathcal{F}}}_{mm',nn'} \left( \boldsymbol{\check{\mathcal{F}}}_{mm',nn'} \right)^{\rm T} \mathbf{\check{\Psi}}^{uu'} \right\} P_{\rm s}\\
&+\left(\Re\left\{ \mathcal{F}_{mm',nn'}\left[\frac{\alpha M}{2}\right] \right\}\right)^2 P_{\rm s} \delta_{uu'},
\end{align}
where
\begin{align}
\nonumber
\mathbf{\check{\Psi}}^{uu'}
\triangleq&\mathbb{E}\left\{ \left( \boldsymbol{\check{\psi}}^{uu'} - \mathbb{E}\left\{ \boldsymbol{\check{\psi}}^{uu'} \right\} \right) \left( \boldsymbol{\check{\psi}}^{uu'} - \mathbb{E}\left\{ \boldsymbol{\check{\psi}}^{uu'} \right\} \right)^{\rm T} \right\}\\
\nonumber
=&\frac{1}{2}\begin{bmatrix}
\Re\left\{ \mathbf{\Xi}^{uu'} + \mathbf{\check{\Xi}}^{uu'} \right\} & \Im\left\{ -\mathbf{\Xi}^{uu'} + \mathbf{\check{\Xi}}^{uu'} \right\} \\
\Im\left\{ \mathbf{\Xi}^{uu'} + \mathbf{\check{\Xi}}^{uu'} \right\} & \Re\left\{ \mathbf{\Xi}^{uu'} - \mathbf{\check{\Xi}}^{uu'} \right\}
\end{bmatrix}
\end{align}
represents the interference brought by the equalization error, and the entries of $\mathbf{\Xi}^{uu'}$ and $\mathbf{\check{\Xi}}^{uu'}$ in the $l$-th row and $l'$-th column are $\varepsilon_{ll'}^{uu'}$ and $\check{\varepsilon}_{ll'}^{uu'}$, respectively. The expressions of $\varepsilon_{ll'}^{uu'}$ and $\check{\varepsilon}_{ll'}^{uu'}$ are derived in Appendix B in supplementary file. It can be seen from the calculation that the orders of magnitude of the elements in $\mathbf{\check{\Psi}}^{uu'}$ are about $10^{-8} \sim 10^{-6}$ when $M=256$, $N_{\rm r}=16$ under the channel models and parameter setting in Section \ref{sim}.Therefore, the equalization error can be ignored, which implies that the interference caused by the frequency-selective channels are almost eliminated. In this case, the main factor impacting the SINR performance comes from the noise. By increasing the transmit SNR, $\text{SINR}_{m,n}^u$ can approach the theoretical SIR upper bound, i.e., the SIR in the case of SISO when the channel response is unit impulse, which is expressed as
\begin{align}\label{UB}
\text{SIR}_{m,n} = \frac{\bar{P}_{mm,nn}} {\mathop{\sum_{n'=-\infty}^{+\infty} \sum_{m'=0}^{M-1}}\limits_{(m',n')\neq(m,n)} \bar{P}_{mm',nn'}},
\end{align}
where $\bar{P}_{mm',nn'} \triangleq \left(\Re\left\{ \mathcal{F}_{mm',nn'}\left[\frac{\alpha M}{2}\right] \right\}\right)^2 P_{\rm s}$. $\bar{P}_{{\rm z},m,n}^u$ in (\ref{SINR}) is the average of $P_{{\rm z},m,n}^u$ over different channel realizations, where $P_{{\rm z},m,n}^u$ denotes the power of the noise $\hat{z}_{m,n}^u$ contained in $\hat{s}_{m,n}^u$. The expression of $\bar{P}_{{\rm z},m,n}^u$ is derived in Appendix C in supplementary file. To make the conclusion more reliable, we will confirm the theoretical SINR in (\ref{SINR}) through simulations in the subsequent section.

\subsection{Complexity Analysis}\label{CA}
In this subsection, we take the equalizers designed under ZF criterion as an example to analyze the complexity of calculating the equalizer coefficients of the following equalization schemes:

1) \textbf{Single-tap}: The scheme using the single-tap equalizers;

2) \textbf{Single-tap\&PDP}: The scheme using the single-tap equalizers assisted by the low-rate PDP equalizers designed by \cite[Eq. (21)]{AFFB18};

3) \textbf{Multi-tap}: The scheme using the low-rate multi-tap equalizers designed by \cite[Eq. (18)]{TAJM11};

4) \textbf{Two-stage}: The scheme using the low-rate multi-tap equalizers designed by the proposed two-stage framework.

We assume that instantaneous CSI in frequency domain is known at the BS, and high-rate equalizer length is $L_{\rm g}=M$. The conclusions are also applicable to the case of MMSE criterion. For simplicity, we evaluate the complexity in terms of the number of complex multiplications.

For \textbf{Single-tap} scheme, the main computational cost comes from matrix multiplication and matrix inversion, which are $\mathcal{O}(MN_{\rm t}^2N_{\rm r})$ and $\mathcal{O}(MN_{\rm t}^3)$, respectively. Generally, since $N_{\rm r} \geq N_{\rm t}$, the computational complexity is $\mathcal{O}(MN_{\rm t}^2N_{\rm r})$. For \textbf{Single-tap\&PDP} scheme, we assume that high-rate PDP equalizers are designed by FS method. The computational complexity contains two parts. One is designing high-rate PDP equalizers, and the complexity mainly comes from fast Fourier transform (FFT) and inverse FFT (IFFT), which is $\mathcal{O}(M^2N_{\rm t}\log_2M)$. The other is transforming the high-rate into low-rate PDP equalizers, and the complexity mainly comes from the convolution with band-pass filters, which is $\mathcal{O}(L'_{\rm g}M^2N_{\rm t})$. Therefore, the total computational complexity is $\mathcal{O}(MN_{\rm t}(N_{\rm t}N_{\rm r} + M\log_2M + L'_{\rm g}M))$. For \textbf{Multi-tap} scheme, the main computational cost comes from matrix multiplication and matrix inversion, and the total computational complexity is $\mathcal{O}(L'_{\rm g}MN_{\rm t}N_{\rm r}(N_{\rm t} + L'_{\rm g}))$.

For \textbf{Two-stage} scheme, in Stage-1, the main computational cost comes from matrix multiplication, matrix inversion, and IFFT. The complexity of these computations are $\mathcal{O}(MN_{\rm t}^2N_{\rm r})$, $\mathcal{O}(MN_{\rm t}^3)$, and $\mathcal{O}(MN_{\rm t}N_{\rm r}\log_2 M)$, respectively. In Stage-2, we first need to compute $\boldsymbol{\hat{e}}_m^{r,u}$, whose complexity is $\mathcal{O}( \frac{M L_{\rm f}}{D_1})$. The complexity of the multiplication between $\mathbf{\bar{F}}_m$ and $\boldsymbol{\hat{e}}_m^{r,u}$ is $\mathcal{O}(L'_{\rm g}( \frac{L_{\rm f}}{D_1} +L'_{\rm g}))$. Thus, the total computational complexity is $\mathcal{O}(\frac{MN_{\rm t}N_{\rm r}(ML_{\rm f}+L'_{\rm g}(L_{\rm f}+D_1L'_{\rm g}))}{D_1} + MN_{\rm t}^2N_{\rm r})$. Table \ref{tab:compare} summarizes the above complexity analysis and compares the schemes in terms of the required kinds of CSI and equalizers. It can be seen that the complexity of the proposed scheme is a polynomial of degree one respect to the number of BS antennas, which is affordable in the case of massive MIMO. Besides, \textbf{Sinle-tap\&PDP} scheme requires two kinds of CSI and equalizers, which complicates the channel estimation.

\newcommand{\tabincell}[2]{\begin{tabular}{@{}#1@{}}#2\end{tabular}}
\begin{table*}
\caption{Complexity of calculating the equalizer coefficients and required kinds of CSI and equalizers}
\renewcommand{\arraystretch}{1}
\begin{center}
{\begin{tabular}{|c|c|c|c|c|}\hline\label{tab:compare}
\diagbox[width=13em]{Indexes}{Schemes} & \textbf{Single-tap} & \textbf{Single-tap\&PDP} & \textbf{Multi-tap} & \textbf{Two-stage} \\
\hline
Complexity in $\mathcal{O}(\cdot)$ notation & $MN_{\rm t}^2N_{\rm r}$ & \tabincell{c}{$MN_{\rm t}(N_{\rm t}N_{\rm r} + M\log_2M$ \\ $ + L'_{\rm g}M)$} & $L'_{\rm g}MN_{\rm t}N_{\rm r}(N_{\rm t} + L'_{\rm g})$ & \tabincell{c}{$\frac{MN_{\rm t}N_{\rm r}(ML_{\rm f}+L'_{\rm g}(L_{\rm f}+D_1L'_{\rm g}))}{D_1}$ \\ $ + MN_{\rm t}^2N_{\rm r}$} \\
\hline
Required kinds of CSI & \tabincell{c}{One kind:\\ Instantaneous CSI} & \tabincell{c}{Two kinds:\\ Instantaneous and\\ statistical CSI} & \tabincell{c}{One kind:\\ Instantaneous CSI} & \tabincell{c}{One kind:\\ Instantaneous CSI} \\
\hline
Required kinds of equalizers & \tabincell{c}
{One kind:\\ Single-tap equalizers} & \tabincell{c}{Two kinds:\\ Single-tap equalizers and\\ multi-tap equalizers} & \tabincell{c}{One kind:\\ Multi-tap equalizers} & \tabincell{c}{One kind:\\ Multi-tap equalizers} \\
\hline
\end{tabular}}
\end{center}
\end{table*}

\section{Simulation Results}\label{sim}
\begin{table}[h]
\caption{Parameters setup for simulation}
\renewcommand{\arraystretch}{1}
\begin{center}
{\begin{tabular}{|c|c|c|}\hline\label{tab:all}
\textbf{Categories} & \textbf{Parameters} & \textbf{Values} \\\hline
\multirow{4}*{General setup} & Number of users $N_{\rm t}$ & 8 \\
\cline{2-3}
&Number of subcarriers $M$ & 256 \\
\cline{2-3}
&Subcarrier spacing & 30kHz\\
\cline{2-3}
&Bandwidth & 7.68MHz\\
\hline

\multirow{8}*{Channel} & \multirow{4}*{Channel model} & User-$1$, $2$: 3GPP-EVA \\
\cline{3-3}
& & User-$3$, $4$: 3GPP-ETU \\
\cline{3-3}
& & User-$5$, $6$: ITU-Ped.A \\
\cline{3-3}
& & User-$7$, $8$: ITU-Ped.B \\
\cline{2-3}

& \multirow{4}*{Channel length $L_{\rm h}$} & 3GPP-EVA\cite{3GPP}: 28 \\
\cline{3-3}
& & 3GPP-ETU\cite{3GPP}: 47 \\
\cline{3-3}
& & ITU-Ped.A\cite{ITU}: 12 \\
\cline{3-3}
& & ITU-Ped.B\cite{ITU}: 37 \\
\hline

\multirow{2}*{Prototype filter} & Type & PHYDYAS\cite{MB10} \\
\cline{2-3}
& Overlapping factor $\kappa$ & 4 \\
\hline

\multirow{2}*{Equalizer} & High-rate equalizer length $L_{\rm g}$ & $M$ \\
\cline{2-3}
& Delay term $\alpha$ & 1 \\
\hline


\end{tabular}}{}
\end{center}
\end{table}

In this section, we provide simulation results to evaluate the performance of the proposed equalization framework. For the FBMC/OQAM-based massive MIMO system in all simulations, the simulation parameters are summarised in Table \ref{tab:all} unless stated otherwise. Assume that all users share the same transmit SNR, which is denoted by ${\gamma}=P_{\rm s}/\sigma_{\rm z}^2$. In the sequence, we compare SIR, SINR and average mean squared error (MSE) performance (each performance are averaged over different channel realizations) of the equalization schemes illustrated in Section \ref{CA}, where the equalizers are designed under ZF or MMSE criterion.

\begin{figure}
\centering
\includegraphics[width=3.5in]{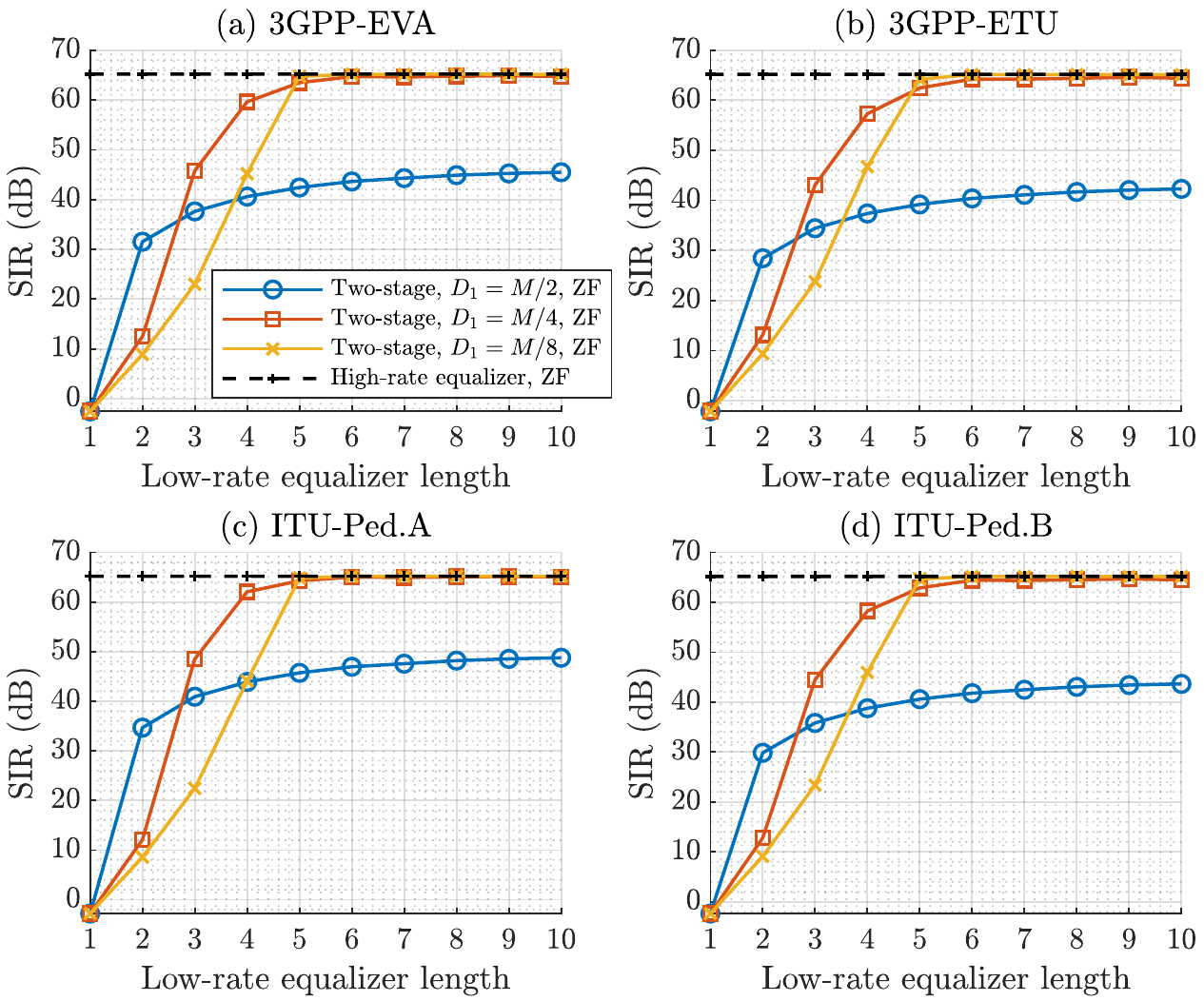}
\centering
\caption{SIR performance at a certain subcarrier versus low-rate equalizer length under different channel models. Three sets of $D_1$ are selected, where $D_1=\frac{M}{2}$ implies that the receiver works under the conventional one-step decimation. The performance of the high-rate equalizers designed in Stage-1 is added as a benchmark. All the equalizers are designed under ZF criterion. $N_{\rm r}=16$, $\gamma=10$ dB.}
\label{SIRvsLeq_N=16_SNR=10}
\end{figure}

\begin{figure}
\centering
\includegraphics[width=3.5in]{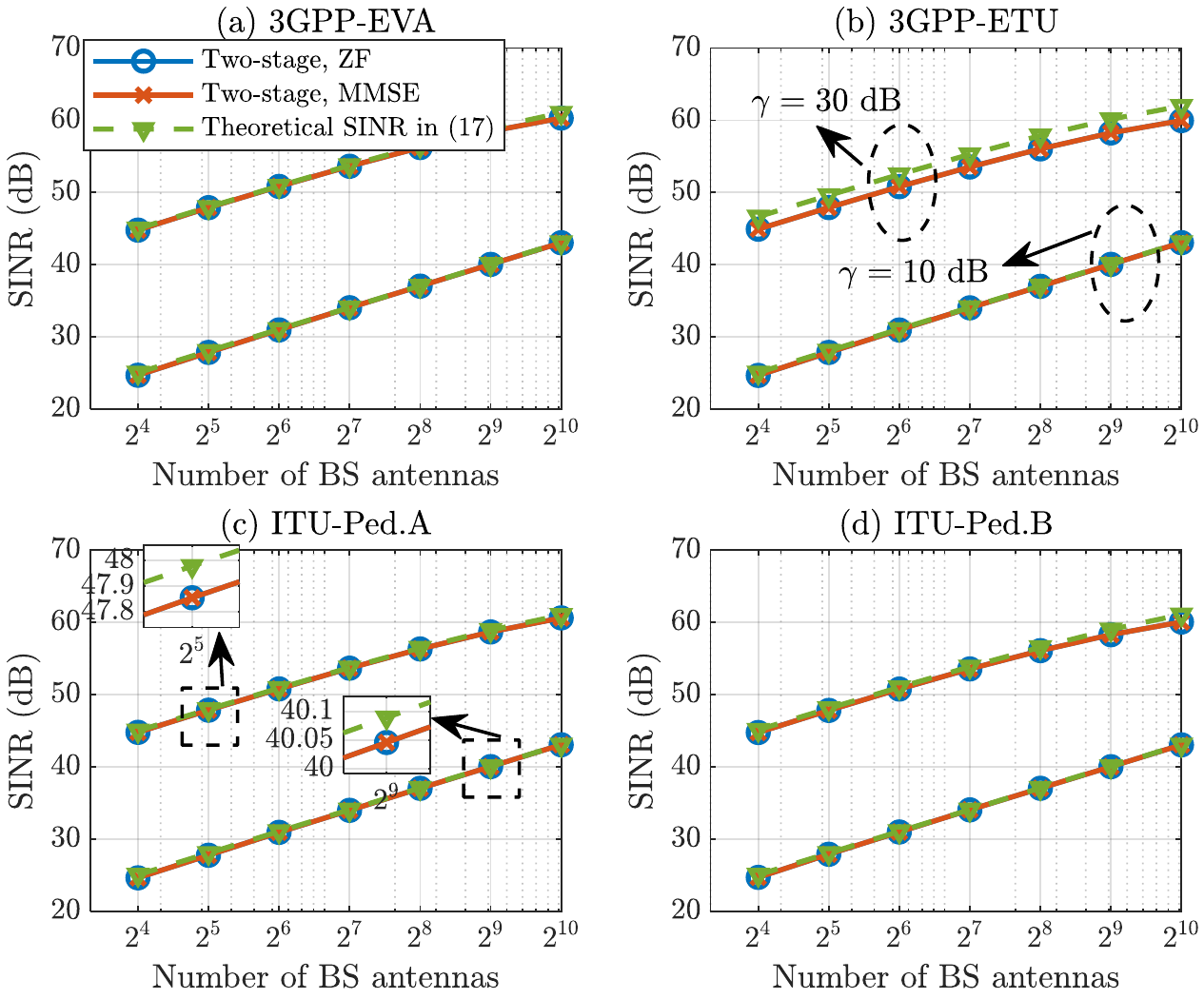}
\centering
\caption{SINR performance at a certain subcarrier versus the number of BS antennas under different channel models. Both ZF and MMSE criteria are considered. The theoretical SINR in (\ref{SINR}) is added as a benchmark. $D_1=\frac{M}{4}$, $L'_{\rm g}=5$, $\gamma=10$ or $30$ dB.}
\label{SINRvsN_theo}
\end{figure}

\subsection{Performance Under MISO Case}
In this subsection, we consider the single-user case, where different channel models are also taken into account. First, we use SIR performance to evaluate the performance of the low-rate equalizers using our proposed transform approach, and the impact of noise is excluded. Fig. \ref{SIRvsLeq_N=16_SNR=10} shows the SIR performance at a certain subcarrier versus the equalizer length $L'_{\rm g}$. It can be seen that the SIR increases along with the equalizer length and tends to a saturation level when $L'_{\rm g} \geq 5$. However, in the case of $D_1=\frac{M}{4}$ and $D_1=\frac{M}{8}$, the saturation level approaches the SIR performance of the corresponding high-rate equalizers. This indicates that by using our proposed two-step decimation combined with LS method in Stage-2, the low-rate equalizers have almost the same performance with the high-rate equalizers designed in Stage-1. Besides, although the emissions of $f_m^*[-l]$ outside the range $[\frac{2\pi(m-1)}{M},\frac{2\pi(m+1)}{M})$ are tiny enough, they still have a great impact on the accuracy of the low-rate equalizers. By using the two-step decimation, the problem can be effectively solved.

Then, we verify the validity of the analyzed theoretical SINR performance in (\ref{SINR}). Fig. \ref{SINRvsN_theo} shows the SINR performance at a certain subcarrier versus the number of BS antennas $N_{\rm r}$, where both ZF and MMSE criteria are considered. This figure confirms that the theoretical SINR values match the simulated ones regardless of the channel model, which verifies the derivation in \ref{SP}. Besides, the equalizers designed under ZF and MMSE have almost the same performance. Thus, in the following subsections, equalizers are all based on ZF criterion.

\subsection{SINR Performance Under MIMO Case}
\begin{figure}
\centering
\includegraphics[width=3.5in]{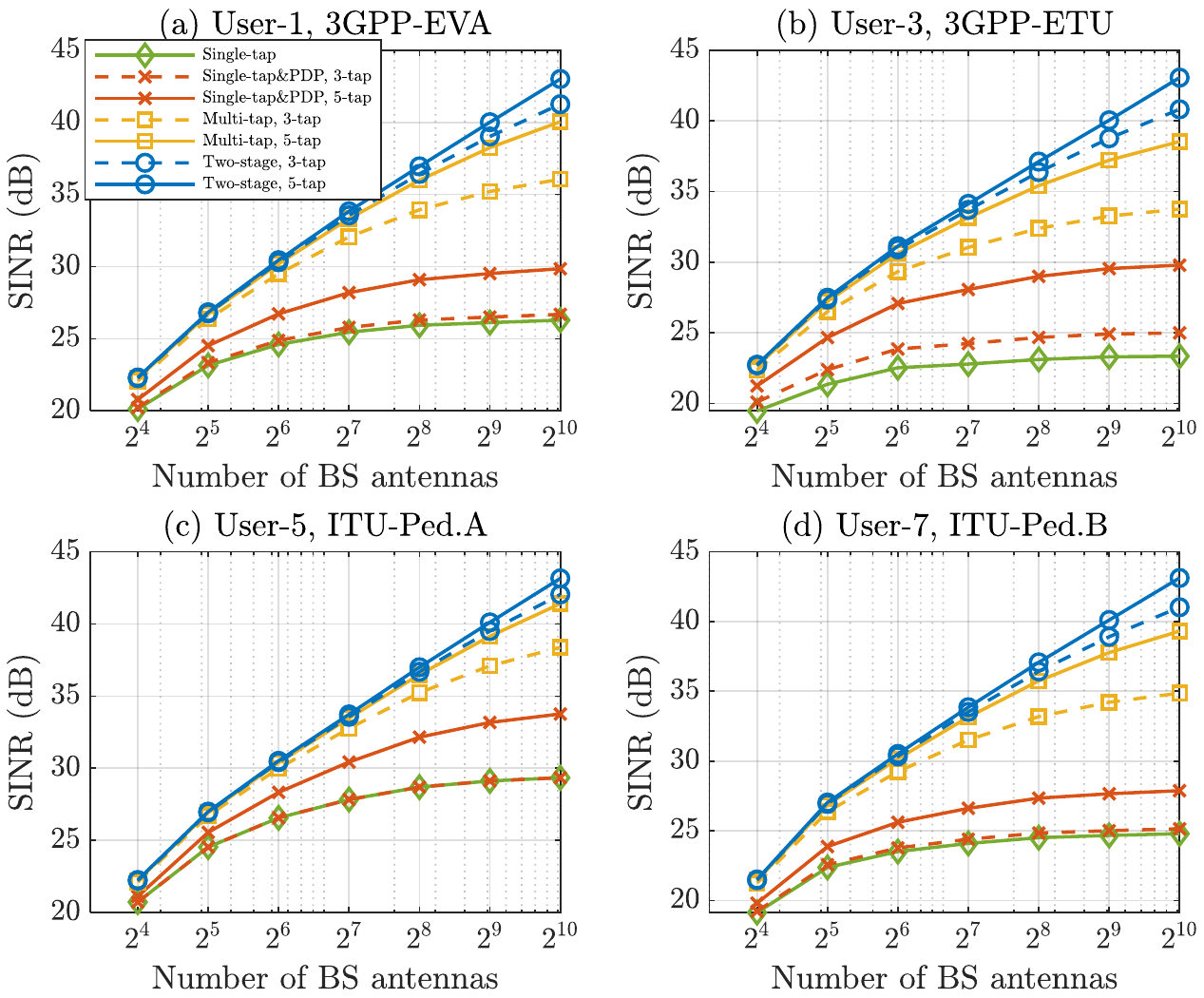}
\centering
\caption{SINR performance of various equalization schemes at a certain subcarrier versus the number of BS antennas. Two sets of $L'_{\rm g}$ are selected. For \textbf{Two-stage} scheme, $D_1=\frac{M}{4}$. $\gamma=10$ dB.}
\label{SINRvsN_SNR=10}
\end{figure}

In this subsection, the multi-user case is considered, where the channel model corresponding to each user is shown in Table \ref{tab:all}. We evaluate SINR performance of various equalization schemes at a certain subcarrier under different channel models. Fig. \ref{SINRvsN_SNR=10} shows SINR performance as a function of different numbers of BS antennas, where the transmit SNR is $\gamma=10$ dB. It can be observed that: 1) For the single-tap equalization, as $N_{\rm r}$ increases, SINR performance tends to be limited. This validates the claim in \cite{AFFB18} that by using single-tap equalizers, some residual interference from the same user remains even with infinite number of BS antennas. By using the proposed scheme, this problem can be effectively solved, and the SINR performance continues to increase with the the number of BS antennas; 2) For a certain number of BS antennas, the proposed scheme outperforms other schemes, which is valid even when $L'_{\rm g}=3$ for the proposed scheme while $L'_{\rm g}=5$ for other schemes; 3) For \textbf{Single-tap\&PDP} scheme, although $N_{\rm r}$ may be large enough (e.g. $N_{\rm r}=1024$) to characterize the residual interference as channel PDPs, the low-rate PDP equalizers perform not well for three reasons due to \textbf{R1}-\textbf{R3} illustrated in Section \ref{SE}; 4) There is no obvious impact on \textbf{Two-stage} scheme under different channel models.

Fig. \ref{SINRvsSNR_N=16} and Fig. \ref{SINRvsSNR_N=64} demonstrates the SINR performance as a function of different transmit SNRs, where the number of BS antennas is $N_{\rm r}=16$ and $N_{\rm r}=64$, respectively. We can see that: 1) When $\gamma> 40$dB, the SINR performance of the proposed scheme with $L'_{\rm g}=5$ can continue to grow until it approaches the SIR upper bound in (\ref{UB}), while other performance curves tend to be limited under much lower values. It illustrates that in this case, the main factor impacting the SINR performance is not the interference brought by the channels but the noise. Through increasing the transmit SNR, even with finite number of BS antennas (e.g., $N_{\rm r}=64$), the SINR performance can also approach the theoretical SIR upper bound, which is the property that other schemes do not have; 2) For \textbf{Multi-tap} scheme, due to the loss of channel information caused by FBMC demodulation, its performance is degraded obviously. Fortunately, this can be effectively avoided by Stage-1 of \textbf{Two-stage} scheme.

\begin{figure}
\centering
\includegraphics[width=3.5in]{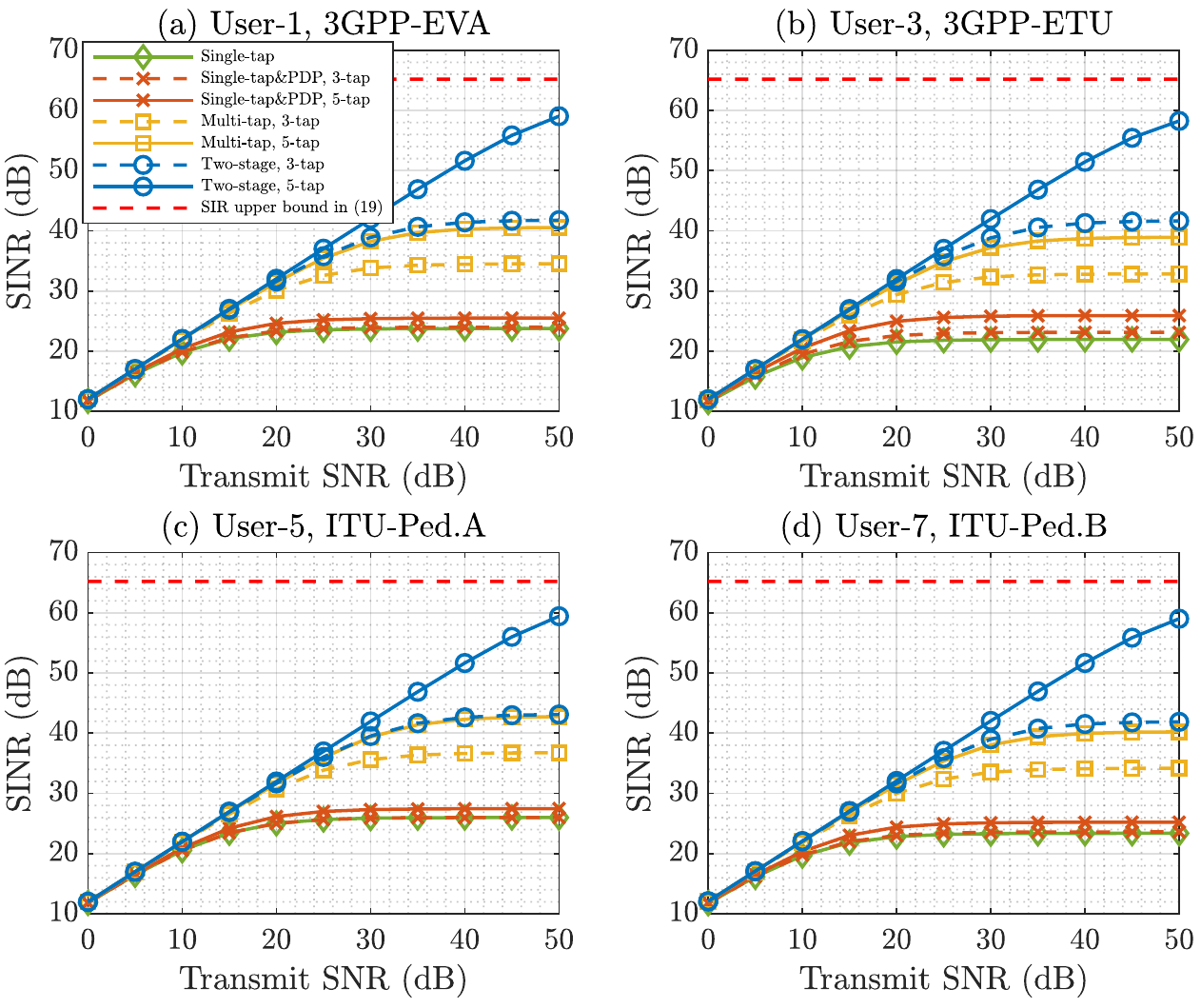}
\centering
\caption{SINR performance of various equalization schemes at a certain subcarrier versus transmit SNR under different channel models. Two sets of $L'_{\rm g}$ are selected. The SIR upper bound in (\ref{UB}) is added as a benchmark. For \textbf{Two-stage} scheme, $D_1=\frac{M}{4}$. $N_{\rm r}=16$.}
\label{SINRvsSNR_N=16}
\end{figure}

\begin{figure}
\centering
\includegraphics[width=3.5in]{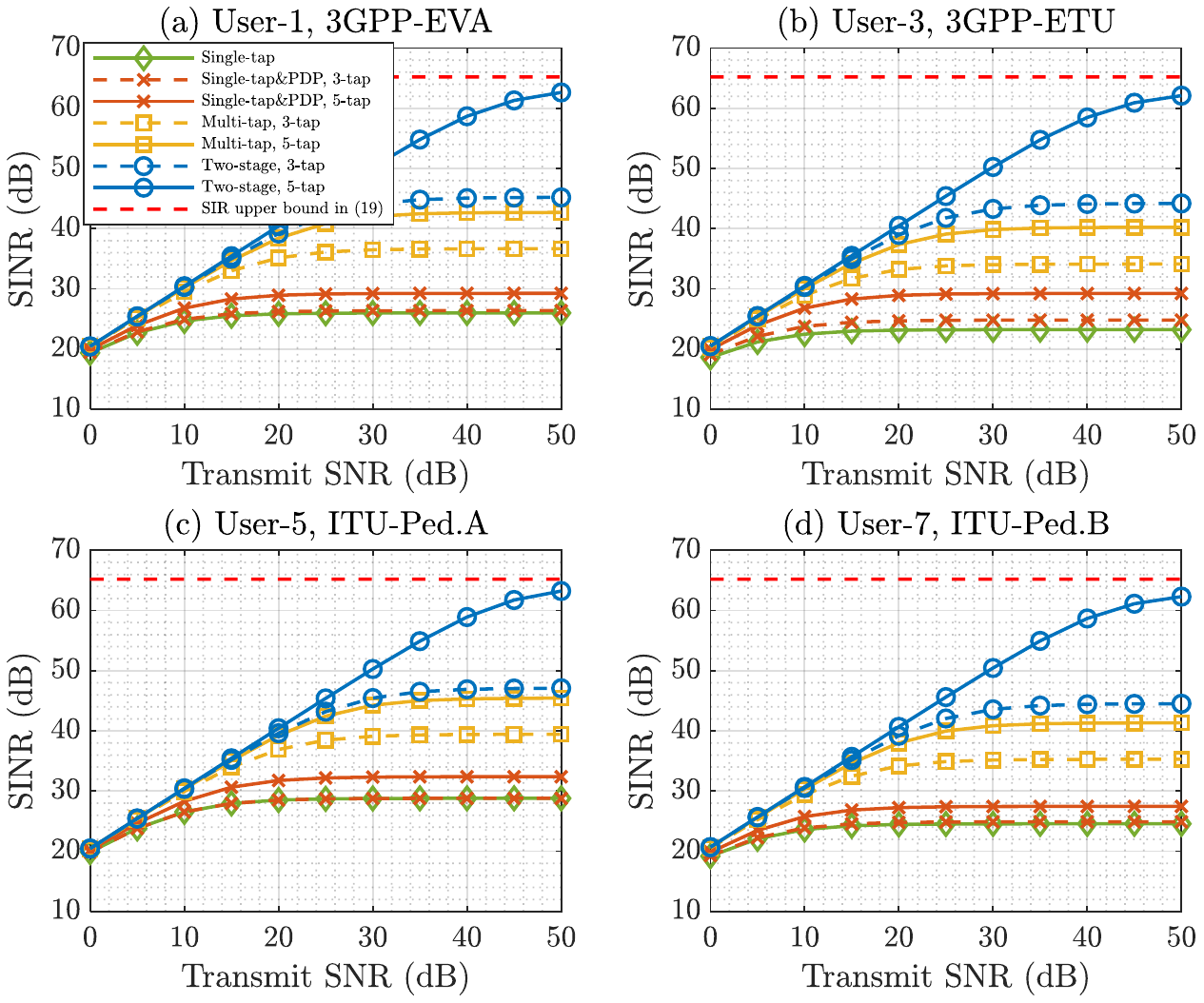}
\centering
\caption{SINR performance of various equalization schemes at a certain subcarrier versus transmit SNR. Two sets of $L'_{\rm g}$ are selected. The SIR upper bound in (\ref{UB}) is added as a benchmark. For \textbf{Two-stage} scheme, $D_1=\frac{M}{4}$. $N_{\rm r}=64$.}
\label{SINRvsSNR_N=64}
\end{figure}

\subsection{Average MSE Performance Under MIMO Case}
In this subsection, the multi-user case is considered, where the channel model corresponding to each user is shown in Table \ref{tab:all}. We evaluate the performance of various equalization schemes in terms of average MSE between transmitted and estimated OQAM symbols, which is defined as
\begin{align}
\text{MSE} \triangleq \mathbb{E}\left\{ \frac{\sum_{n=0}^{N_{\rm d}-1} \sum_{m=0}^{M-1} \Vert \boldsymbol{\hat{s}}_{m,n} - \boldsymbol{s}_{m,n} \Vert^2}{N_{\rm t}MN_{\rm d}} \right\},
\end{align}
where $N_{\rm d}$ is the number of OQAM symbols per subcarrier. Moreover, we also provide the performance in presence of imperfect CSI brought by the channel estimation which is training-based linear MMSE channel estimation method \cite{PGAK20}. The estimate of the channel matrix at the $m$-th subcarrier in frequency domain between the BS and the users is
\begin{align}
\mathbf{\hat{\tilde{H}}}_m = \frac{P_{\rm p} \mathbf{\tilde{H}}_m + \mathbf{\tilde{Z}}_m}{P_{\rm p} + \sigma_{\rm z}^2},
\end{align}
where $\mathbf{\tilde{H}}_m \triangleq \sum_{\ell=0}^{L_{\rm h}-1}{\mathbf{H}\left[ \ell\right] e^{-j\frac{2\pi m\ell}{M}}}$. $\mathbf{\tilde{Z}}_m$ can be regarded as the estimated error matrix whose elements are independent random variables and identically distributed as $\mathcal{CN}(0, P_{\rm p}\sigma_{\rm z}^2)$. $P_{\rm p}\triangleq 2P_sL_{\rm p}$ represents the pilot power, and $L_{\rm p}$ is the number of training symbols per subcarrier. The related parameter setting are shown in Table \ref{tab:mse}.

\begin{table}[h]
\caption{Parameters setup for MSE performance simulation}
\renewcommand{\arraystretch}{1}
\begin{center}
{\begin{tabular}{|c|c|}\hline\label{tab:mse}
\textbf{Parameters} & \textbf{Values} \\\hline
Number of OQAM symbols per subcarrier $N_{\rm d}$ & 96 \\
\cline{1-2}
Number of training symbols per subcarrier $L_{\rm p}$ & 8 \\
\cline{1-2}
Modulation format & 16-QAM\\
\hline

\end{tabular}}{}
\end{center}
\end{table}

Fig. \ref{MSEvsSNR_16QAM_N=16} and \ref{MSEvsSNR_16QAM_N=64} show the MSE performance as a function of different transmit SNRs, where the number of BS antennas is $N_{\rm r}=16$ and $N_{\rm r}=64$, respectively. We also evaluate the performance of CP-OFDM-based massive MIMO system as a benchmark. Simulations under perfect and imperfect CSI are labelled as `P-CSI' and `I-CSI', respectively. For the FBMC-based systems, as the figures show: 1) Compered with other equalization schemes, the proposed two-stage scheme has better performance in terms of high transmit SNR. Furthermore, we nearly achieve the same performance with CP-OFDM, where the channel frequency response is completely flat over each individual subcarrier; 2) When the number of BS antennas increases from $16$ to $64$, the performance of the two-stage scheme increases significantly, while that of other schemes has only minor improvements; 3) The imperfect CSI brought by the channel estimation has an impact on all equalization schemes in terms of low transmit SNR, which is about $1.8$ dB. However, when the MSE performance converges along with the increase of transmit SNR, the impact of the imperfect CSI will gradually disappear. The reason is that all the schemes are based on linear equalization. When the transmit SNR is large enough, noise is no longer the main impact factor, and the equalization error caused by the imperfect CSI also tends to zero. Therefore, these schemes have similar robustness in presence of imperfect CSI.

\begin{figure}
\centering
\includegraphics[width=3.5in]{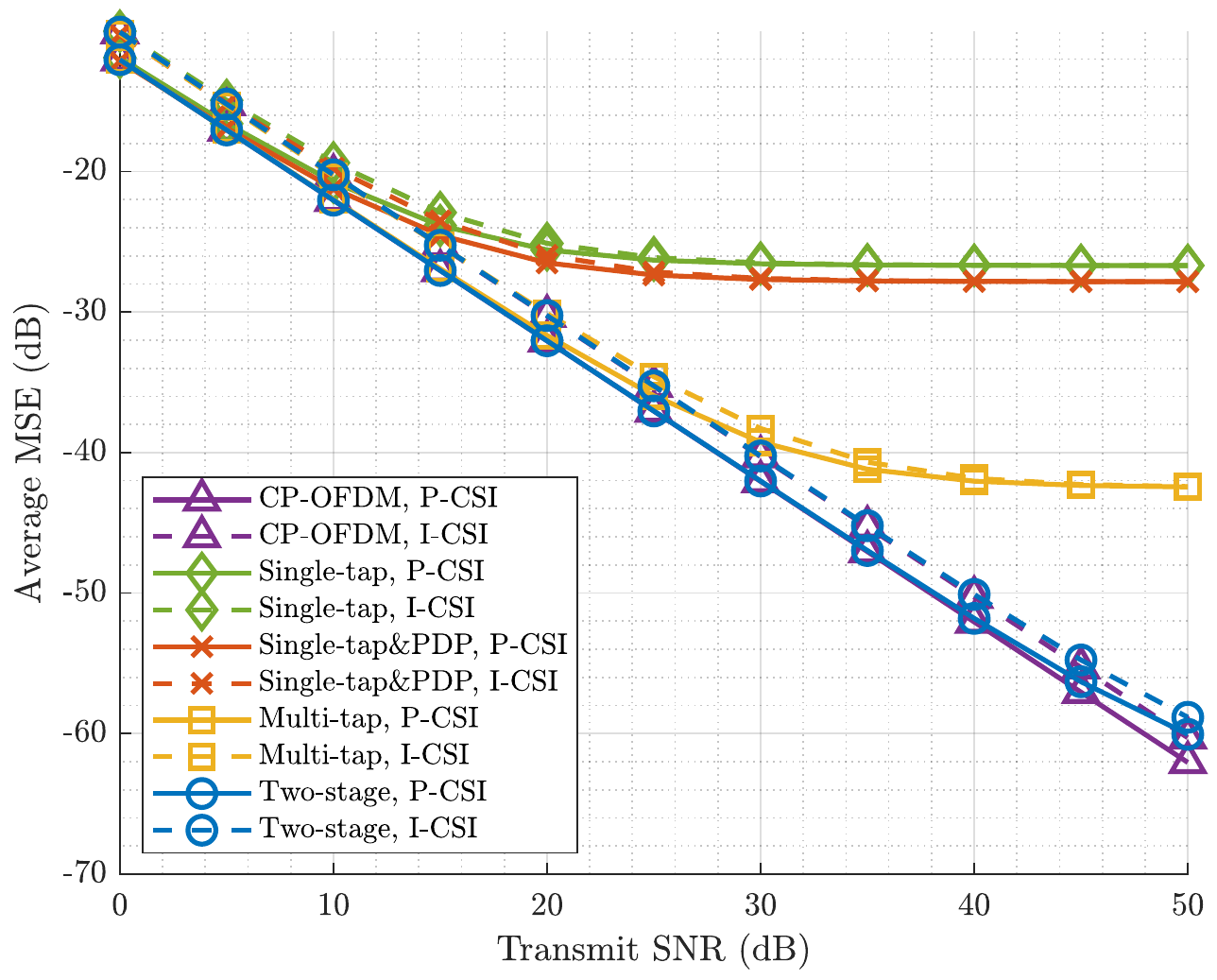}
\centering
\caption{MSE performance of different equalization schemes versus transmit SNR. Simulations under perfect and imperfect CSI are labelled as `P-CSI' and `I-CSI', respectively. For \textbf{Two-stage} scheme, $D_1=\frac{M}{4}$. $L'_{\rm g}=5$, $N_{\rm r}=16$.}
\label{MSEvsSNR_16QAM_N=16}
\end{figure}

\begin{figure}
\centering
\includegraphics[width=3.5in]{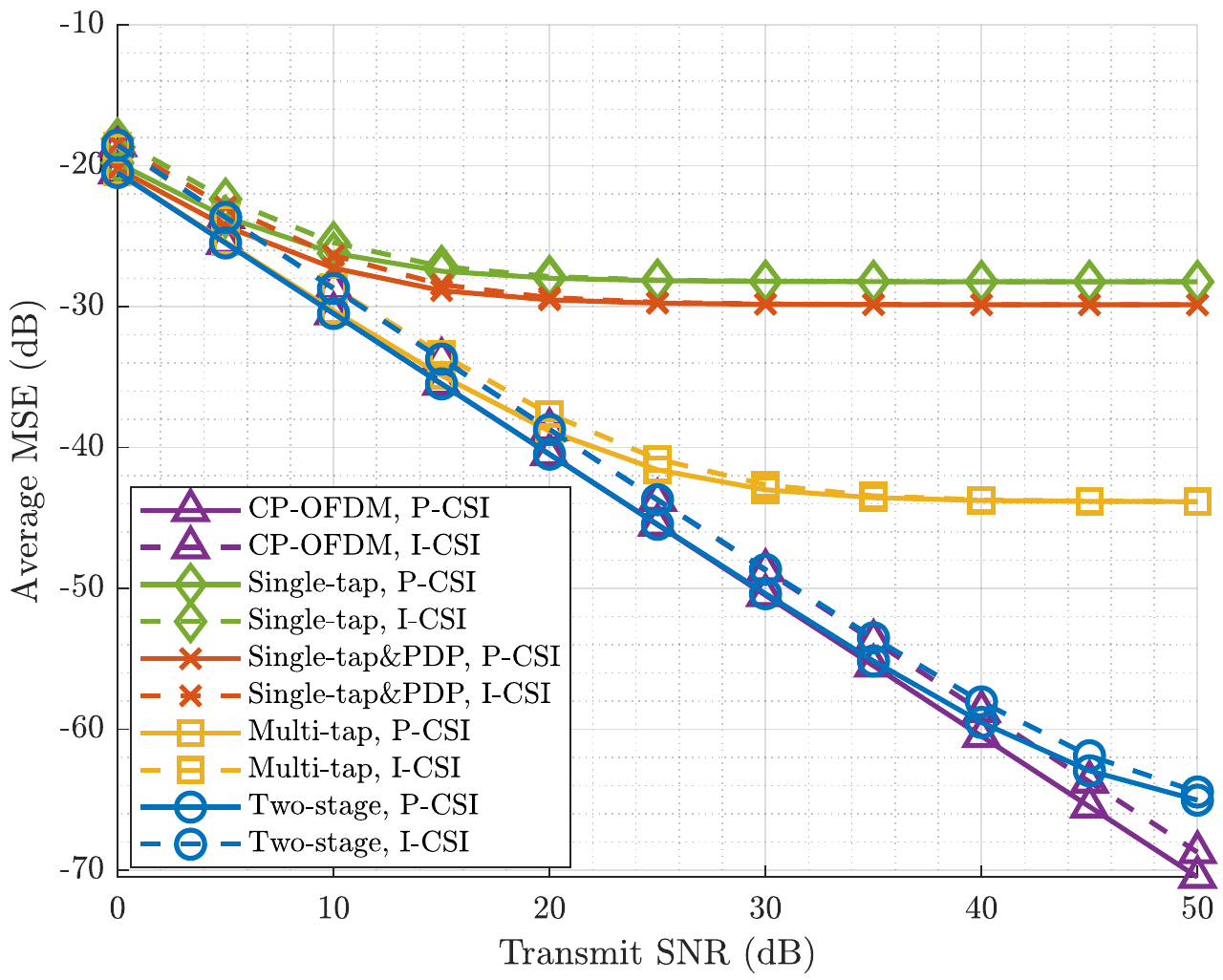}
\centering
\caption{MSE performance of different equalization schemes versus transmit SNR. Simulations under perfect and imperfect CSI are labelled as `P-CSI' and `I-CSI', respectively. For \textbf{Two-stage} scheme, $D_1=\frac{M}{4}$. $L'_{\rm g}=5$, $N_{\rm r}=64$.}
\label{MSEvsSNR_16QAM_N=64}
\end{figure}

\section{Conclusion}
In this paper, we propose a two-stage design scheme of equalizers for uplink FBMC/OQAM-based massive MIMO system in highly frequency-selective channels, which overcomes the shortcomings of the prior equalization schemes. Besides, we derive the theoretical expression for SINR performance of the proposed scheme, which shows that the interference caused by the channels can be almost eliminated even with finite number of BS antennas. In this case, by increasing the transmit SNR, the SINR value can approach the SIR upper bound. Simulation results validate the above analyses and show that our scheme outperforms some prior works. Notice that this scheme can also be applied to other forms of FBMC and new waveforms based massive MIMO systems. In the future works, we will consider the downlink scenario and the multi-cell system.
\appendices

\section{Proof of Proposition \ref{(a)=(c)}}\label{Proof(a)=(c)}
This proof is divided into two parts. First, we prove that the receivers shown in Fig. \ref{Method1&2Model}-(a) and Fig. \ref{Method1&2Model}-(c) are equivalent, where $\bar{g}_m^{r,u}[n]$ is obtained by \textbf{Method 1}. Then, we prove that through \textbf{Method 1} and \textbf{Method 2}, the same $\bar{g}_m^{r,u}[n]$ can be obtained. For a more general conclusion, here, the coefficient $\frac{M}{2}$ is replaced by $D$.

\emph{Part 1}: The DTFT of $\check{g}_{{\rm (1)},m}^{r,u}[l]$ can be expressed as
\begin{align}
\tilde{\check{g}}_{{\rm (1)},m}^{r,u}(\omega)=
\left\{ \begin{matrix}
\tilde{g}^{r,u}(\omega), &\omega \in [\frac{\pi(m-1)}{D},\frac{\pi(m+1)}{D}) \\
0, &\text{else}.
\end{matrix}
\right.
\end{align}
Since $\bar{g}_m^{r,u}[n] = D\left(\check{g}_{{\rm (1)},m}^{r,u}[l]\right)_{\downarrow D}$, the frequency response of $\bar{g}_m^{r,u}[n]$ is \cite{MJV14}
\begin{align}\label{tildegmru1}
\tilde{\bar{g}}_m^{r,u}(\omega) = \sum_{i=0}^{\frac{M}{2}-1} \tilde{\check{g}}_{{\rm (1)},m}^{r,u}\left( \frac{\omega-2\pi i}{D} \right).
\end{align}
Therefore, the DTFT of the output signal $\hat{y}_m^{r,u}[n]$ in Fig. \ref{Method1&2Model}-(c) can be expressed as
\begin{align}\label{tildeymru}
\nonumber
\tilde{\hat{y}}_{{\rm (c)},m}^{r,u}(\omega)=
&\frac{1}{D} \sum_{i=0}^{D-1} \sum_{i'=0}^{D-1} \tilde{y}^r\left(\frac{\omega-2\pi i}{D}\right) \tilde{f}_m^*\left(\frac{\omega-2\pi i}{D}\right)\\
&\times \tilde{\check{g}}_{{\rm (1)},m}^{r,u}\left(\frac{\omega-2\pi i'}{D}\right),
\end{align}
where $\tilde{f}_m^*(\omega)$ is the DTFT of the analysis filter $f_m^*[-l]$. Due to the fact that both $\tilde{f}_m^*(\omega)$ and $\tilde{\check{g}}_{{\rm (1)},m}^{r,u}(\omega)$ are band-limited to $[\frac{\pi(m-1)}{D},\frac{\pi(m+1)}{D})$, it can be found that $\tilde{f}_m^*\left(\frac{\omega-2\pi i}{D}\right)$ and $\tilde{\check{g}}_{{\rm (1)},m}^{r,u}\left(\frac{\omega-2\pi i'}{D}\right)$ are band-limited to $[(2i+m-1)\pi,(2i+m+1)\pi)$ and $[(2i'+m-1)\pi,(2i'+m+1)\pi)$, respectively. When $i\neq i'$, there is no intersection between those two ranges. Hence,
\begin{align}\label{ineqi}
\tilde{f}_m^*\left(\frac{\omega-2\pi i}{D}\right) \tilde{\check{g}}_{{\rm (1)},m}^{r,u}\left(\frac{\omega-2\pi i'}{D}\right) = 0, \quad i\neq i'.
\end{align}
Moreover, note that $\tilde{f}_m^*(\omega) \tilde{\check{g}}_{{\rm (1)},m}^{r,u}(\omega) = \tilde{f}_m^*(\omega) \tilde{g}^{r,u}(\omega)$. (\ref{tildeymru}) can be simplified as
\begin{align}
\nonumber
\tilde{\hat{y}}_{{\rm (c)},m}^{r,u}(\omega)=
&\frac{1}{D} \sum_{i=0}^{D-1} \tilde{y}^r\left(\frac{\omega-2\pi i}{D}\right) \tilde{f}_m^*\left(\frac{\omega-2\pi i}{D}\right)\\
&\times \tilde{g}^{r,u}\left(\frac{\omega-2\pi i}{D}\right),
\end{align}
which is equivalent to the DTFT of the output signal shown in Fig. \ref{Method1&2Model}-(a). Thus, the first part of the proof is completed.

\emph{Part 2}: Note that $\tilde{\check{g}}_{{\rm (2)},m}^{r,u}(\omega)$ can be written as $\tilde{\check{g}}_{{\rm (2)},m}^{r,u}(\omega) = \sum_{\ell=0}^{D-1} \tilde{\check{g}}_{{\rm (1)},m}^{r,u}\left(\omega-\frac{2\pi\ell}{D}\right)$. Thus, the DTFT of $\bar{g}_m^{r,u}[n]$ obtained by \textbf{Method 2} can be expressed as
\begin{align}\label{tildegmru2}
\tilde{\bar{g}}_m^{r,u}(\omega)
= \frac{1}{D} \sum_{i=0}^{D-1} \sum_{\ell=0}^{D-1} \tilde{\check{g}}_{{\rm (1)},m}^{r,u}\left(\frac{\omega-2\pi i-2\pi\ell}{D}\right).
\end{align}
Due to the periodicity of DTFT, for $\forall i\in \{0,\cdots,D-1\}$, it follows
\begin{align}
\nonumber
&\sum_{\ell=0}^{D-1} \tilde{\check{g}}_{{\rm (1)},m}^{r,u}\left(\frac{\omega-2\pi i-2\pi\ell}{D}\right)\\
=&\sum_{\ell=0}^{D-1} \tilde{\check{g}}_{{\rm (1)},m}^{r,u}\left(\frac{\omega-2\pi (i+\ell-Dv_{\ell})}{D}\right),
\end{align}
where $v_{\ell}, \ell=0,\cdots,D-1$ are arbitrary integers. Obviously, when $\ell$ evolves from $0$ to $D-1$, there always exists a corresponding $v_{\ell}$, such that $i'=i+\ell-Dv_{\ell}$ can also evolve from $0$ to $D-1$. In this way, (\ref{tildegmru2}) can be transformed into $\tilde{\bar{g}}_m^{r,u}(\omega) = \sum_{i'=0}^{D-1} \tilde{\check{g}}_{{\rm (1)},m}^{r,u}\left(\frac{\omega-2\pi i'}{D}\right)$, which is the same as (\ref{tildegmru1}). Therefore, the second part of the proof is completed.

\section{Calculation of The Statistics of $\psi^{uu'}[l]$ }\label{psistatistic}
The equivalent channel $\mathbf{H}_{\rm eq}[l]$ can be calculated as
\begin{align}
\nonumber
\mathbf{H}_{\rm eq}[l]
=&\mathbf{G}[l] \star \mathbf{H}[l]\\
\nonumber
=&\frac{1}{M} \sum_{\ell=l-(M-1)}^{l} \sum_{m=0}^{M-1} \boldsymbol{\mathcal{W}}_m^{\rm H} \mathbf{H}[\ell] e^{j\frac{2\pi m(l-\ell-\alpha M/2)}{M}},
\end{align}
where $\boldsymbol{\mathcal{W}}_m = \mathbf{\tilde{H}}_m \left( \mathbf{\tilde{H}}_m^{\rm H} \mathbf{\tilde{H}}_m \right)^{-1}$. It can be divided into the following cases according to the value of $l$:

\emph{Case 1}: When $L_{\rm h}-1 \leq l \leq M-1$, $\mathbf{H}_{\rm eq}[l]$ can be simplified as
\begin{align}
\nonumber
\mathbf{H}_{\rm eq}[l]
=\frac{1}{M} \sum_{m=0}^{M-1} \boldsymbol{\mathcal{W}}_m^{\rm H} \mathbf{\tilde{H}}_m e^{j\frac{2\pi m(l-\alpha M/2)}{M}}
=\mathbf{\Delta}\left[l-\frac{\alpha M}{2}\right].
\end{align}
In this case, $\mathbf{\Psi}[l]=\mathbf{H}_{\rm eq}[l]-\mathbf{\Delta}\left[l-\frac{\alpha M}{2}\right]=\mathbf{0}$. It means that when $L_{\rm h}-1 \leq l \leq M-1$, there is no equalization error.

\emph{Case 2}: When $0\leq l < L_{\rm h}-1$, $\mathbf{H}_{\rm eq}[l]$ can not be simplified as that in \emph{Case 1}. Thus, we turn to calculating the statistics of $\psi^{uu'}[l]$. The mean of $\psi^{uu'}[l]$ can be calculated as
\begin{align}
\nonumber
\mathbb{E}\left\{\psi^{uu'}[l]\right\} =&\mathbb{E}\left\{ \left[ \mathbf{H}_{\rm eq}[l] - \mathbf{\Delta}\left[l-\frac{\alpha M}{2}\right] \right]_{u,u'} \right\}\\
\nonumber
=&\frac{1}{M^2} \sum_{\ell=0}^{l} \sum_{m=0}^{M-1} \sum_{m'=0}^{M-1} \mathbb{E} \left\{ \left( \boldsymbol{w}_m^u \right)^{\rm H} \boldsymbol{\tilde{h}}_{m'}^{u'} \right\} e^{j\frac{2\pi m'\ell}{M}}\\
\nonumber
&\times e^{j\frac{2\pi m(l-\ell-\alpha M/2)}{M}}\\
\nonumber
\overset{\text{(a)}}{=}&\frac{1}{M^2} \sum_{\ell=0}^{l} \sum_{m=0}^{M-1} \sum_{m'=0}^{M-1} \sum_{\ell'=0}^{L_{\rm h}-1} \delta_{uu'} q^{u'}[\ell'] e^{j\frac{2\pi (m-m')\ell'}{M}}\\
\nonumber
&\times e^{j\frac{2\pi m'\ell}{M}} e^{j\frac{2\pi m(l-\ell-\alpha M/2)}{M}}\\
\nonumber
=&\sum_{\ell=0}^{l} \delta_{uu'} q^{u'}[\ell] \delta\left[l-\frac{\alpha M}{2}\right]=0,
\end{align}
where $\boldsymbol{w}_m^u$ and $\boldsymbol{\tilde{h}}_{m}^{u}$ are the $u$-th column of $\boldsymbol{\mathcal{W}}_m$ and $\mathbf{\tilde{H}}_m$, respectively. (a) follows from the fact that $\boldsymbol{\tilde{h}}_{m'}^{u'}$ can be transformed into a combination of two terms, i.e., $\boldsymbol{\tilde{h}}_{m'}^{u'} = \tau_{mm'}^{u'} \boldsymbol{\tilde{h}}_{m}^{u'} + \boldsymbol{\tilde{h}}_{mm'}^{u',{\rm in}}$,
where $\boldsymbol{\tilde{h}}_{mm'}^{u',{\rm in}}$ is independent of $\boldsymbol{\tilde{h}}_{m}^{u'}$ and the correlation coefficient $\tau_{mm'}^{u'}$ is
\begin{align}
\nonumber
\tau_{mm'}^{u'} = \mathbb{E} \left\{ h_{m'}^{r,u'} \left( h_{m}^{r,u'} \right)^* \right\} = \sum_{l=0}^{L_{\rm h}-1} q^{u'}[l] e^{j\frac{2\pi (m-m')l}{M}},
\end{align}
which satisfies that $\left( \tau_{mm'}^{u'} \right)^* = \tau_{m'm}^{u'}$. Besides, it can be obtained that $\left(\boldsymbol{w}_m^u\right)^{\rm H} \boldsymbol{\tilde{h}}_m^{u'}=\delta_{uu'}$. The second order statistic of $\psi^{uu'}[l]$ can be calculated as
\begin{align}\label{varep}
\nonumber
&\mathrel{\phantom{=}}\varepsilon_{ll'}^{uu'}\\
\nonumber
&\triangleq\mathbb{E}\left\{ \psi^{uu'}[l] \left( \psi^{uu'}[l'] \right)^* \right\}\\
\nonumber
&=\frac{1}{M^2} \sum_{\ell =0}^l \sum_{m=0}^{M-1} \sum_{\ell '=0}^{l'} \sum_{m'=0}^{M-1} \mathbb{E}\left\{ \left( \boldsymbol{w}_{m}^{u} \right)^{\rm H} \boldsymbol{h}^{u'}[\ell] \left( \boldsymbol{h}^{u'}[\ell'] \right)^{\rm H} \boldsymbol{w}_{m'}^{u} \right\}\\
&\mathrel{\phantom{=}}\times e^{j\frac{2\pi m(l-\ell -\alpha M/2)}{M}} e^{-j\frac{2\pi m'(l'-\ell'-\alpha M/2)}{M}}.
\end{align}
When $u=u'$, $\mathbb{E}\left\{ \left( \boldsymbol{w}_{m}^{u} \right)^{\rm H} \boldsymbol{h}^{u}[\ell] \left( \boldsymbol{h}^{u}[\ell'] \right)^{\rm H} \boldsymbol{w}_{m'}^{u} \right\}$ can be calculated as
\begin{align}
\nonumber
&\mathrel{\phantom{=}}\mathbb{E}\left\{ \left( \boldsymbol{w}_{m}^{u} \right)^{\rm H} \boldsymbol{h}^{u}[\ell] \left( \boldsymbol{h}^{u}[\ell'] \right)^{\rm H} \boldsymbol{w}_{m'}^{u} \right\}\\
\nonumber
&=\frac{1}{M^2}\sum_{m''=0}^{M-1} \sum_{m'''=0}^{M-1} \mathbb{E}\left\{ \left( \boldsymbol{w}_{m}^{u} \right)^{\rm H}
\boldsymbol{\tilde{h}}_{mm''}^{u,\mathrm{in}} \left( \boldsymbol{\tilde{h}}_{m'm'''}^{u,\mathrm{in}} \right)^{\rm H} \boldsymbol{w}_{m'}^{u} \right\}\\
\nonumber
&\mathrel{\phantom{=}}\times e^{j\frac{2\pi m''\ell}{M}} e^{-j\frac{2\pi m'''\ell'}{M}}
+q^{u}[\ell]q^{u}[\ell'] e^{j\frac{2\pi (m\ell-m'\ell')}{M}}.
\end{align}
Since $\mathbb{E}\left\{ \left( \boldsymbol{w}_{m}^{u} \right)^{\rm H}
\boldsymbol{\tilde{h}}_{mm''}^{u,\mathrm{in}} \left( \boldsymbol{\tilde{h}}_{m'm'''}^{u,\mathrm{in}} \right)^{\rm H} \boldsymbol{w}_{m'}^{u} \right\}$ is hard to calculate, here, we consider the case of multiple-input single-output (MISO) to simplify the problem and provide an insight into its value. Thus, it can be further calculated as
\begin{align}
\nonumber
&\mathbb{E}\left\{ \left( \boldsymbol{w}_{m} \right)^{\rm H}
\boldsymbol{\tilde{h}}_{mm''}^{\mathrm{in}} \left( \boldsymbol{\tilde{h}}_{m'm'''}^{\mathrm{in}} \right)^{\rm H} \boldsymbol{w}_{m'} \right\}\\
\nonumber
\approx&\frac{ \tau_{m'''m''} - \tau_{m'''m'} \tau_{m'm''} - \tau_{m'''m} \tau_{mm''} + \tau_{m'''m'} \tau_{m'm} \tau_{mm''}} {\tau_{m'm} N_{\mathrm{r}}},
\end{align}
where the superscript $u$ is omitted. We denote $O_{mm'}$ as (\ref{Omm'}) shown at the top of the next page.
\begin{figure*}[!t]
\normalsize
\begin{align}\label{Omm'}
\nonumber
O_{mm'}\triangleq
&\frac{1}{M^2} \sum_{m''=0}^{M-1} \sum_{m'''=0}^{M-1} \frac{ \tau_{m'''m''} - \tau_{m'''m'} \tau_{m'm''} - \tau_{m'''m} \tau_{mm''} + \tau_{m'''m'} \tau_{m'm} \tau_{mm''}} {\tau_{m'm}} e^{j\frac{2\pi m''\ell}{M}} e^{-j\frac{2\pi m'''\ell'}{M}}\\
=&\frac{q[\ell]\delta[\ell-\ell'] - q[\ell]q[\ell'] \left( e^{j\frac{2\pi m(\ell-\ell')}{M}} + e^{j\frac{2\pi m'(\ell-\ell')}{M}} \right)}{\tau_{m'm}} + q[\ell]q[\ell'] e^{j\frac{2\pi (m\ell-m'\ell')}{M}}
\end{align}
\hrulefill
\vspace*{4pt}
\end{figure*}
Taking the above results into (\ref{varep}), we can obtain that when $u=u'$, $\varepsilon_{ll'}^{uu'}$ can be approximately expressed as
\begin{align}
\nonumber
\varepsilon_{ll'}^{uu'}
\approx&\frac{1}{M^2N_{\rm r}} \sum_{\ell =0}^l \sum_{m=0}^{M-1} \sum_{\ell '=0}^{l'} \sum_{m'=0}^{M-1} O_{mm'}\\
\nonumber
&\times e^{j\frac{2\pi m(l-\ell -\alpha M/2)}{M}} e^{-j\frac{2\pi m'(l'-\ell'-\alpha M/2)}{M}}.
\end{align}

Similarly, when $u\neq u'$, $\varepsilon_{ll'}^{uu'}$ can be approximately expressed as
\begin{align}
\nonumber
\varepsilon_{ll'}^{uu'}
\approx&\frac{1}{M^2N_{\mathrm{r}}} \sum_{\ell=0}^{l_{\rm min}} \sum_{m=0}^{M-1} \sum_{m'=0}^{M-1} \frac{q^{u'}[\ell]}{\tau_{m'm}}\\
\nonumber
&\times e^{j\frac{2\pi m\left( l-\ell -\alpha M/2 \right)}{M}} e^{-j\frac{2\pi m'\left( l'-\ell -\alpha M/2 \right)}{M}}.
\end{align}

Following the similar derivation, another second order statistic of $\psi^{uu'}[l]$ is
\begin{align}
\nonumber
\check{\varepsilon}_{ll'}^{uu'}
\triangleq&\mathbb{E}\left\{ \psi^{uu'}[l] \psi^{uu'}[l'] \right\}\\
\nonumber
\approx&\left\{ \begin{matrix}
\frac{1}{M^2N_{\rm r}} \sum_{\ell =0}^l \sum_{m=0}^{M-1} \sum_{\ell '=0}^{l'} \sum_{m'=0}^{M-1} \check{O}_{mm'}\\
\times e^{j\frac{2\pi m(l-\ell -\alpha M/2)}{M}} e^{j\frac{2\pi m'(l'-\ell'-\alpha M/2)}{M}}, &u=u' \\
0, &u\neq u'
\end{matrix}
\right.
\end{align}
where $\check{O}_{mm'}$ is denoted by (\ref{checkOmm'}) at the top of the next page.
\begin{figure*}[!t]
\normalsize
\begin{align}\label{checkOmm'}
\nonumber
\check{O}_{mm'}
\triangleq&\frac{1}{M^2} \sum_{m''=0}^{M-1} \sum_{m'''=0}^{M-1} \frac{\left( \tau_{mm'''} - \tau_{m'm'''} \tau_{mm'} \right) \left( \tau_{m'm''} - \tau_{m'm} \tau_{mm''} \right)} {\vert\tau_{mm'}\vert^2} e^{j\frac{2\pi m''\ell}{M}} e^{j\frac{2\pi m'''\ell'}{M}}\\
=&\frac{ \left( e^{j\frac{2\pi (m\ell'+m'\ell)}{M}} - \tau_{m'm}e^{j\frac{2\pi m(\ell+\ell')}{M}} - \tau_{mm'}e^{j\frac{2\pi m'(\ell+\ell')}{M}} \right) q[\ell] q[\ell']} {\vert\tau_{mm'}\vert^2} + q[\ell] q[\ell'] e^{j\frac{2\pi (m\ell+m'\ell')}{M}}
\end{align}
\hrulefill
\vspace*{4pt}
\end{figure*}

\emph{Case 3}: When $M-1 < l < M+L_{\rm h}-1$, the statistics of $\psi^{uu'}[l]$ are
\begin{align}
\nonumber
&\mathbb{E}\left\{\psi^{uu'}[l]\right\} =0,\\
\nonumber
&\varepsilon_{ll'}^{uu'}\approx\\
\nonumber
&\left\{ \begin{matrix}
\frac{1}{M^2N_{\rm r}} \sum_{\ell=l-(M-1)}^{L_{\rm h}-1} \sum_{m=0}^{M-1} \sum_{\ell'=l'-(M-1)}^{L_{\rm h}-1} \sum_{m'=0}^{M-1}\\
O_{mm'} e^{j\frac{2\pi m(l-\ell -\alpha M/2)}{M}} e^{-j\frac{2\pi m'(l'-\ell'-\alpha M/2)}{M}}, &u=u'\\
\frac{1}{M^2N_{\mathrm{r}}} \sum_{\ell=l_{\rm max}-(M-1)}^{L_{\rm h}-1} \sum_{m=0}^{M-1} \sum_{m'=0}^{M-1}\\
\frac{q^{u'}[\ell]}{\tau_{m'm}} e^{j\frac{2\pi m\left( l-\ell -\alpha M/2 \right)}{M}} e^{-j\frac{2\pi m'\left( l'-\ell -\alpha M/2 \right)}{M}}, &u\neq u',
\end{matrix}
\right.\\
\nonumber
&\check{\varepsilon}_{ll'}^{uu'}\approx\\
\nonumber
&\left\{ \begin{matrix}
\frac{1}{M^2N_{\rm r}} \sum_{\ell=l-(M-1)}^{L_{\rm h}-1} \sum_{m=0}^{M-1} \sum_{\ell'=l'-(M-1)}^{L_{\rm h}-1} \sum_{m'=0}^{M-1}\\
\check{O}_{mm'} e^{j\frac{2\pi m(l-\ell -\alpha M/2)}{M}} e^{j\frac{2\pi m'(l'-\ell'-\alpha M/2)}{M}}, &u=u' \\
0, &u\neq u',
\end{matrix}
\right.
\end{align}
where $l_{\rm max}$ denotes the maximum of $l$ and $l'$.

\emph{Case 4}: When $0 \leq l < L_{\rm h}-1$, $M-1 < l' < M+L_{\rm h}-1$, the second statistics of $\psi^{uu'}[l]$ are
\begin{align}
\nonumber
&\varepsilon_{ll'}^{uu'}\approx\\
\nonumber
&\left\{ \begin{matrix}
\frac{1}{M^2N_{\rm r}} \sum_{\ell=0}^{l} \sum_{m=0}^{M-1} \sum_{\ell'=l'-(M-1)}^{L_{\rm h}-1} \sum_{m'=0}^{M-1}\\
O_{mm'} e^{j\frac{2\pi m(l-\ell -\alpha M/2)}{M}} e^{-j\frac{2\pi m'(l'-\ell'-\alpha M/2)}{M}}, &u=u'\\
\frac{1}{M^2N_{\mathrm{r}}} \sum_{\ell=l'-(M-1)}^{l} \sum_{m=0}^{M-1} \sum_{m'=0}^{M-1}\\
\frac{q^{u'}[\ell]}{\tau_{m'm}} e^{j\frac{2\pi m\left( l-\ell -\alpha M/2 \right)}{M}} e^{-j\frac{2\pi m'\left( l'-\ell -\alpha M/2 \right)}{M}}, &\mathop{u\neq u'}\limits_{l'-l \leq M-1} \\
0, & \mathop{u\neq u'}\limits_{l'-l > M-1}
\end{matrix}
\right.\\
\nonumber
&\check{\varepsilon}_{ll'}^{uu'}\approx\\
\nonumber
&\left\{ \begin{matrix}
\frac{1}{M^2N_{\rm r}} \sum_{\ell=0}^{l} \sum_{m=0}^{M-1} \sum_{\ell'=l'-(M-1)}^{L_{\rm h}-1} \sum_{m'=0}^{M-1}\\
\check{O}_{mm'} e^{j\frac{2\pi m(l-\ell -\alpha M/2)}{M}} e^{j\frac{2\pi m'(l'-\ell'-\alpha M/2)}{M}}, &u=u' \\
0, &u\neq u'.
\end{matrix}
\right.
\end{align}
Conversely, when $0 \leq l' < L_{\rm h}-1$, $M-1 < l < M+L_{\rm h}-1$, the results can be obtained through the conjugate symmetry property of covariance matrix.

\section{Calculation of $\bar{P}_{{\rm z},m,n}^u$ }\label{barPzmnu}
The noise contained in the final estimated data symbol $\hat{s}_{m,n}^u$ can be expressed as
\begin{align}
\nonumber
\hat{z}_{m,n}^u=\sum_{r=0}^{N_{\rm r}-1} \sum_{l=1-L_{\rm f}}^{M-1} \sum_{\ell=0}^{M-1} \Re\left\{ g^{r,u}[\ell] f_m^*[\ell-l] z^r[-l] e^{-j\theta_{m,n}} \right\}.
\end{align}
Thus, its average power can be calculated by
\begin{align}
\nonumber
&\mathrel{\phantom{=}} \bar{P}_{{\rm z},m,n}^u\\
\nonumber
&=\mathbb{E}\left\{ \left(\hat{z}_{m,n}^u\right)^2 \right\}\\
\nonumber
&=\frac{\sigma_{\rm z}^2}{2} \sum_{l=1-L_{\rm f}}^{M-1} \sum_{\ell=0}^{M-1} \sum_{\ell'=0}^{M-1} \sum_{r=0}^{N_{\rm r}-1} \mathbb{E} \left\{ \left(g^{r,u}[\ell]\right)^* g^{r,u}[\ell'] \right\} \\
\nonumber
&\mathrel{\phantom{=}} \times f_m[\ell-l] f_m^*[\ell'-l]\\
\nonumber
&=\frac{\sigma_{\rm z}^2}{2M^2} \sum_{l=1-L_{\rm f}}^{M-1} \sum_{\ell=0}^{M-1} \sum_{\ell'=0}^{M-1} \sum_{m'=0}^{M-1} \sum_{m''=0}^{M-1} \mathbb{E} \left\{ \left(\boldsymbol{w}_{m'}^u\right)^{\rm H} \boldsymbol{w}_{m''}^u \right\} \\
\nonumber
&\mathrel{\phantom{=}} \times f_m[\ell-l] f_m^*[\ell'-l] e^{j\frac{2\pi m'(\ell-\alpha M/2)}{M}} e^{-j\frac{2\pi m''(\ell'-\alpha M/2)}{M}}\\
\nonumber
&\approx\frac{\sigma_{\rm z}^2}{2M^2N_{\rm r}} \sum_{l=1-L_{\rm f}}^{M-1} \sum_{\ell=0}^{M-1} \sum_{\ell'=0}^{M-1} \sum_{m'=0}^{M-1} \sum_{m''=0}^{M-1} \frac{1}{\tau_{m''m'}^u} \\
\nonumber
&\mathrel{\phantom{=}} \times f_m[\ell-l] f_m^*[\ell'-l] e^{j\frac{2\pi m'(\ell-\alpha M/2)}{M}} e^{-j\frac{2\pi m''(\ell'-\alpha M/2)}{M}}.
\end{align}

\end{document}